\documentclass{article}
\usepackage{amsmath}
\usepackage{amsfonts}
\usepackage{amssymb}
\usepackage{graphicx}
\usepackage{color}
\usepackage{caption}
\usepackage{url}
\usepackage[utf8]{inputenc}
\usepackage{enumerate}
\usepackage[round]{natbib} 
\usepackage{authblk}
\providecommand{\U}[1]{\protect\rule{.1in}{.1in}}
\newtheorem{theorem}{Theorem}

\newtheorem{lemma}[theorem]{Lemma}

\newenvironment{proof}[1][Proof]{\noindent\textbf{#1.} }{\ \rule{0.5em}{0.5em}}

\begin{document}
\title{Robust Clustering Using Tau-Scales}
\author[1]{Juan D. Gonzalez} 
\author[2]{Victor J. Yohai} 
\author[3]{Ruben H. Zamar}
\affil[1]{\small Underwater Sound Division. Argentinian Navy Research Office (DIIV) UNIDEF-CONICET.}
\affil[1,2]{\small Instituto de C\'alculo, Universidad de Buenos Aires -
CONICET}
\affil[3]{\small Department of Statistics, University of British Columbia}

\maketitle

\begin{abstract}
K means is a popular   non-parametric clustering procedure introduced by
\citet{steinhaus1956division} and further developed by \citet{macqueen1967some}. It is known,
however, that K means does not perform well in the presence of outliers.
Cuesta-Albertos et al (1997) introduced a robust alternative, trimmed
K means, which can be tuned to be robust or efficient, but cannot achieve
these two properties simultaneously in an adaptive way. To overcome this
limitation we propose a new robust clustering procedure called
K Tau Centers, which is based on the concept of Tau scale introduced by
\citet{yohai1988high}. We show that K Tau Centers  performs well in extensive
simulation studies and real data examples. We also show that the centers found by the proposed method are
consistent estimators of the ``true'' centers defined as the minimizers of the 
the objective function at the population level.
\end{abstract}
\section{Introduction}

\textbf{\ }

Clustering is a useful tool in unsupervised data analysis. Several $p$%
-dimensional measurements, $\mathbf{x}_{1},\mathbf{...,x}_{n},$ are made on $%
n$ items and used to find a number, $K$, of homogeneous groups called 
\textit{clusters}. 

One way to define the clusters is by giving the centers of the clusters.
Suppose that the centers of the clusters $\boldsymbol{\mu}_{1},...%
\boldsymbol{\mu}_{K}$ are given, where each center is an element of $\mathbb{%
R}^{p}.$ Then the clusters $C_{k},$ $1\leq k\leq K$ can be defined by 
\begin{equation}
C_{k}=\{\mathbf{x}_{i}:\min_{1\leq j\leq K}||\mathbf{x}_{i}-\boldsymbol{\mu }%
_{j}||=||\mathbf{x}_{i}-\boldsymbol{\mu}_{k}||\},  \label{one}
\end{equation}
where $\| \cdot \|$ is the Euclidean norm.
In this paper we consider the robust estimation of the cluster  centers.

There are many parametric and nonparametric approaches
for clustering. One of the most popular non-parametric procedures is
\emph{K means}, introduced by \citet{steinhaus1956division},
and popularized \citet{macqueen1967some}. K means 
optimizes a very natural objective function and therefore is  conceptually simple and appealing. Moreover,  K means has been efficiently implemented in statistical software.

The K means clustering procedure can be described as follows. Let 
\begin{equation*}
D(\mathbf{x}_{i,}\boldsymbol{\mu }_{1},...\boldsymbol{\mu }_{K})=\min_{1\leq
j\leq K}||\mathbf{x}_{i}-\boldsymbol{\mu }_{j}||,1\leq i\leq n
\end{equation*}%
Then the centers of the clusters are obtained by 
\begin{equation*}
(\widehat{\boldsymbol{\mu }}_{1},...\widehat{\boldsymbol{\mu }}_{K})=\arg
\min_{\boldsymbol{\mu }_{1},...\boldsymbol{\mu }_{K}}\sum_{i=1}^{n}D^{2}(%
\mathbf{x}_{i,}\boldsymbol{\mu }_{1},...\boldsymbol{\mu }_{K})
\end{equation*}
Unfortunately K-means is very sensitive to the presence of 
outliers, defined as points that lay far away from all the  clusters. To illustrate this point,  in Figure \ref{intro1} (A) we show two clusters generated by two bivariate normal distributions.
The K-mean cluster centers, marked as triangles, are well identified in panel (A). In panel  (B) we add 10\% percent of
outliers and observe that the K-means cluster centers are no longer well identified.

Procedures that are not much affected by the presence of outliers are called
robust. 
\citet{EspaniolesTKMeans} noted that K-means is not robust and
proposed a robust alternative, called trimmed K-means (TK-means).This
procedure find the centers as K-means after eliminating a fraction $\alpha $
of the observations. The trimmed points are iteratively defined as those further away from the 
current cluster centers. K-means is then applied to the remaining points. When the trimming constant $\alpha$ is well specified, the outliers are likely  identified and  trimmed. However,   in practice $\alpha$ is unknown and difficult to estimate. 

\begin{figure}[htb]
\begin{center}
\begin{tabular}{cc}
(A) & (B) \\ 
\includegraphics[scale=0.42]{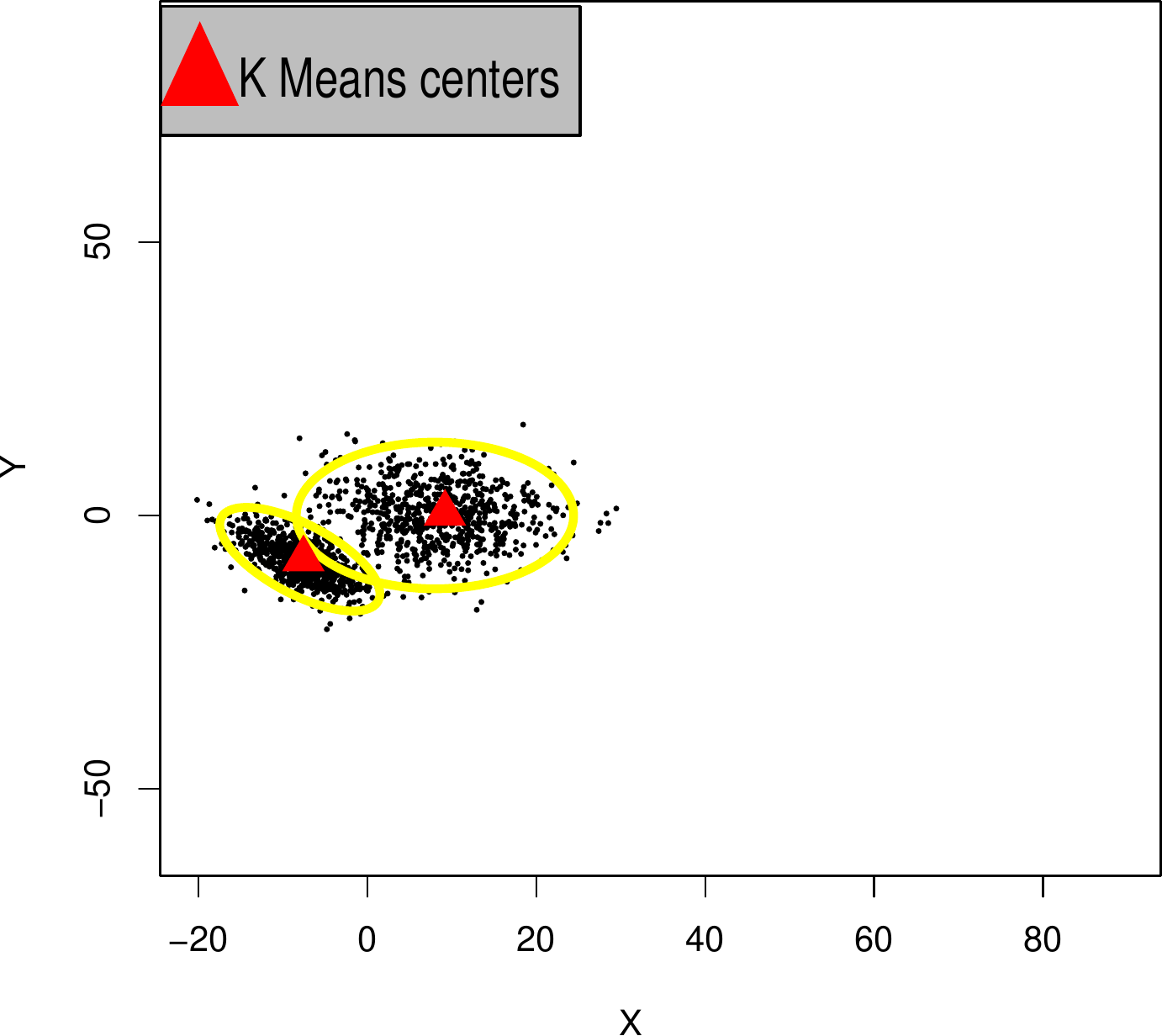} & %
\includegraphics[scale=0.42]{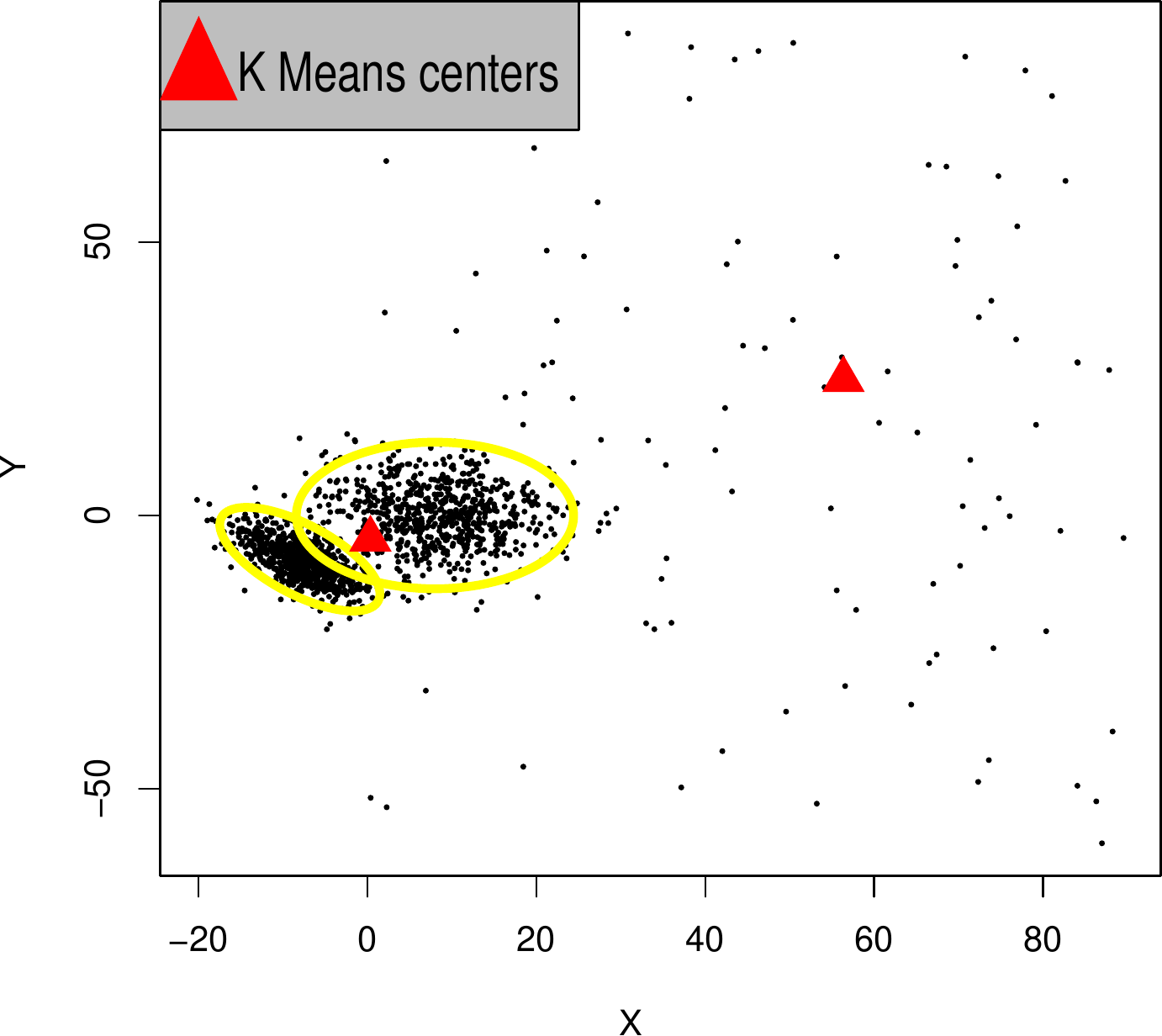}%
\end{tabular}%
\end{center}
\caption{  K means applied to  data before (A) and after (B)  10\% of
outliers are added to the sample }
\label{intro1}
\end{figure}

The rest of the paper is organized as follows. In Section 2 we define the
K-Tau-Centers clustering procedure  (K-Tau) and give an algorithm to 
compute it. In Section 3 we establish the consistency of the estimator.
That is, we show that the estimated cluster centers approach 
the population cluster centers defined as the solution
of  K-Tau  at the population level. In Section 4 we conduct a
simulation study to compare K Tau with K means and TK means. In Section 5 we
apply K-Tau to the processing of satellite images and illustrate the possible
application of robust clustering procedures  to the search of missing objects. 
In Appendix I we derive the estimating equation of K-Tau. In Appendix II we prove the strong consistency of this procedure.
\section{K-Tau-Centers}

Scale estimators play a  central role in the definition of our procedure and therefore they are briefly reviewed in the next section.

\subsection{\textbf{Scale estimators}}

Given a sample $u_{1},...,u_{n}$ of real numbers, a scale estimator $%
s(u_{1},...,u_{n})$ is a measure of how large in absolute value are the
elements of a sample. An scale estimator should have the following
properties: (1)  $s(u_{1},...,u_{n})$ only depend on $|u_{1}|,...,|u_{n}|$; (2) $s(0,...,0)=0$; (3) $|v_{1}|\geq |u_{1}|,...,|v_{n}|\geq |u_{n}|$ implies that  $s(v_{1},...,v_{n})\geq s(u_{1},...,u_{n})$; (4) $s(\lambda u_{1},...,\lambda u_{n})=|\lambda |s(u_{1},...,u_{n})$; (5) $s(u_{i_{1}},...,u_{i_{n}})=s(u_{1},...,u_{n})$ for any permutation $i_{1},...,i_{n}$.
%
%
%
%
The following are examples of scale estimators:

\medskip 
 1) L$_{2}$ scale $s_{2}(u_{1},...,u_{n})=\left( {%
\sum_{i=1}^{n}}u_{i}^{2}/n\right) ^{1/2}$
\medskip 

2)  L$_{1}$ scale $s_{1}(u_{1},...,u_{n})=(1/n){\sum\limits_{i=1}^{n}}%
|u_{i}|$
\medskip 

3)  $\alpha$-trimmed scale $s_{\alpha}(u_{1},...,u_{n}) = 
\sum_{ i \leq n(1-\alpha)}|u_{(i)}|/[ n(1-\alpha)]$
\medskip 

4) Median scale $s_{\text{med}}(u_{1},...,u_{n})=$median$%
(|u_{1}|,...,|u_{n}|)$
\medskip 

5) M scale, Huber (1964). The value $s_{\text{M}%
}(u_{1},...,u_{n})$ is implicitly  defined by a value $s$ that solves 
\begin{equation}
\frac{1}{n}{\displaystyle\sum }_{i=1}^{n}\rho \left({u_{i}}/{s}\right)
=b
\label{sdefff}
\end{equation}%
where $\rho :\mathbb{R}\rightarrow \mathbb{R}_{\geq 0}$ is even, non
decreasing in the absolute value and bounded. Usually $b=\max \rho /2$ \ for
breakdown point equal to $1/2$. The breakdown point is a measure of the
robustness of an estimator. It is the minimum fraction of outliers that may
take the value of the estimator to the boundary of the parameter space. In the case of a scale
estimator the extremes cases  are infinite and zero. The function $\rho $ may be for example a
a member of  the bi-square family given by 
\begin{equation*}
\rho _{_{T}}(u,c)=1-\left( 1-\left({u}/{c}\right)^{2}\right) ^{3}I(|u|\leq c).
\end{equation*}%
\medskip 

6) Tau scale, \citet{yohai1988high}. These scale estimators combine high Gaussian efficiency and high breakdown point. 
This is an advantage over M scales which cannot achieve these two properties simultaneously. To define 
a $\tau $ scale we need two functions, $\rho_1$ and $\rho_{2}$ satisfying the conditions specified in the
definition of M scales. Given a sample $u_{1},...,u_{n}$, we first compute an M scale $%
s=s(u_{1},...,u_{n})$ using $\rho _{1}$. The $\tau$ scale is then defined as
\begin{equation*}
s_{\tau }(u_{1},...,u_{n})=\left(\sum_{i=1}^{n}\rho _{2}\left({%
u_{i}}/{s}\right)/n\right)s.
\end{equation*}

The first two scales are not robust because a single  observation can make
them  arbitrarily large. Scales 3-5 can  achieve high breakdown point and  high efficiency
(but not both properties together). The tau scale  can be tuned to achieve robustness and efficiency simultaneously.

%

\subsection{K-Tau clustering}

First we observe that  K means  can be formulated as a scale minimization problem. In fact, the K means cluster centers are given by 
\begin{equation*}
(\widehat{\boldsymbol{\mu}}_{1},...\widehat{\boldsymbol{\mu}}_{K})=\arg
\min_{\boldsymbol{\mu}_{1},...\boldsymbol{\mu}_{K}}s_{2}(D(\mathbf{x}_{1},%
\boldsymbol{\mu}_{1},...,\boldsymbol{\mu}_{K}),...,D(\mathbf{x}_{n},%
\boldsymbol{\mu}_{1},...,\boldsymbol{\mu}_{K})).
\end{equation*}
It is clear from this formulation that the lack of robustness of K means derives from  the lack of robustness of  $s_{2}$.
Following this reasoning, we define a robust and efficient clustering procedure by replacing  $s_{2}$ by $s_{\tau
}.$ The cluster centers are now defined as: 
\begin{equation*}
(\widetilde{\boldsymbol{\mu }}_{1},...\widetilde{\boldsymbol{\mu }}%
_{K})=\arg \min_{\boldsymbol{\mu }_{1},...\boldsymbol{\mu }_{K}}s_{\tau }(D(%
\mathbf{x}_{1},\boldsymbol{\mu }_{1},...,\boldsymbol{\mu }_{K}),...,D(%
\mathbf{x}_{n},\boldsymbol{\mu }_{1},...,\boldsymbol{\mu }_{K}))
\end{equation*}%
The corresponding cluster partition $\{C_{1},C_{2},...,C_{K}\}$ is given by equation (\ref{one})
with $\boldsymbol{\mu} _{k}=\widetilde{\boldsymbol{\mu }}_{k},$ $\ 1\leq k \leq K.$

\subsection{Estimating equations}

In the case of K means, the clusters centers that minimize $s_{2}$  are simply
the mean of each group. On the other hand, the cluster centers corresponding to  K-TAU
 satisfy the
following fixed point equations 
\begin{equation}
\boldsymbol{\mu }_{k}=\frac{\sum_{i\in C_{k}}w\left({\Vert \mathbf{x}%
_{i}-\boldsymbol{\mu }_{k}\Vert }/{s}\right) \mathbf{x}_{i}}{\sum_{i\in
C_{k}}w\left( {\Vert \mathbf{x}_{i}-\boldsymbol{\mu }_{k}\Vert }/{s}%
\right) }\quad 1\leq k\leq K,  \label{taueq1}
\end{equation}%
where the weight function $w$ is defined as 
\begin{equation}
w(t)=A\frac{\psi _{1}(t)}{t}+B\frac{\psi _{2}(t)}{t},  \label{peso}
\end{equation}%
with  $\psi _{i}=\rho _{i}^{\prime },$ $i=1,2,$ and $A$ and $B$ are data
dependent, namely


\begin{equation}
A=\sum_{i=1}^{n}\left[ 2\rho_{2}\left(d_i/{s}\right) -\psi_{2}\left(d_i/s\right)(d_i/s) \right]
\label{eqDdef}
\end{equation}
\begin{equation}
B=\ \sum_{i=1}^{n}\psi_{1}\left(d_i/{s}\right)(d_i/{s}),  \   \label{eqEdef}
\end{equation}
where $d_i= {D(\mathbf{x}_{i},\boldsymbol{\mu}_{1},...,\boldsymbol{\mu}_{K})}$, and $s$ is the M-scale of all the distances from  each observation to the center
of its  cluster, that is 
\begin{equation}
\frac{1}{n}\sum_{i=1}^{n}\rho_{1}\left( {D(\mathbf{x}_{i},\boldsymbol{%
\mu}_{1},...,\boldsymbol{\mu}_{K})}/{s}\right) =b.  \label{eqSdef}
\end{equation}
The Appendix I contains the derivation of these estimating
equations.

\subsection{Computing algorithm}

The estimating equations given above suggest the following iterative
algorithm to compute the cluster centers.
\medskip


\noindent {\bf Initialization Step:}
$K$ random  points  (or otherwise selected points) are needed
  to start the iterations. 
\medskip

\noindent{\bf Updating Step:} Suppose that centers $\boldsymbol{\mu}%
_{1}^{(l)},\dots,\boldsymbol{\mu}_{K}^{(l)}$  are given. Let $C_{1}^{(l)},...,C_{K}^{(l)}$ be the
corresponding clusters. 
First, we compute the M scale $s^{(l+1)}$  by solving 
\begin{equation*}
\frac{1}{n}\sum_{i=1}^{n}\rho_{1}\left( {D(\mathbf{x}_{i},\boldsymbol{%
\mu}_{1}^{(l)},...,\boldsymbol{\mu}_{K}^{(l)})}/{s}\right) =0.5.
\end{equation*}
Next, the constants $A^{(l+1)}$, $B^{(l+1)}$ are computed using equations (%
\ref{eqDdef}) and (\ref{eqEdef}), replacing $(\boldsymbol{\mu }_{1},...,%
\boldsymbol{\mu }_{K})$ and $s$ by $\boldsymbol{\mu }_{1}^{(l+1)},...,%
\boldsymbol{\mu }_{K}^{(l+1)}$ and $s^{(l+1)\text{ }}$ respectively. The new
weight function $\ w^{(l+1)}$ is then defined by equation (\ref{peso}) 
replacing $A$ and $B$ by $A^{(l+1)}$ and $B^{(l+1)}$. Finally, the new centers
are given by 
\begin{equation*}
\boldsymbol{\mu }_{k}^{(l+1)}=\frac{\sum_{i\in C_{k}}w^{(l+1)}\left( {%
\Vert \mathbf{x}_{i}-\boldsymbol{\mu }_{k}^{(l)}\Vert }/{s^{(l+1)}}\right) 
\mathbf{x}_{i}}{\sum_{i\in C_{k}}w^{(l+1)}\left({\Vert \mathbf{x}_{i}-%
\boldsymbol{\mu }_{k}^{(l)}\Vert }/{s^{(l+1)}}\right) },1\leq k\leq K.
\end{equation*}
\medskip 

\noindent \textbf{Stopping Rule:} Given a prescribed tolerance value $\varepsilon $,
the algorithm stops when
\begin{equation*}
{\left\Vert \boldsymbol{\mu }_{k}^{(l+1)}-\boldsymbol{\mu }%
_{k}^{(l)}\right\Vert }/{\left\Vert \boldsymbol{\mu }_{k}^{(l)}\right\Vert }%
\leq \varepsilon ,\ 1\leq k \leq K.
\end{equation*}
Notice that when $\rho_{1}(t)=\rho_{2}(t)=t^{2}$ and $b=1$ we have $s_{\tau}=s_2$ and  the presented computing
algorithm  reduces to the classic Lloyd's K means Algorithm \citet%
{lloyd1982least}.
\medskip

\noindent {\bf Further details:}
We repeat the above procedure $H$ times with different
 starting points. 
One possibility is to choose the initial $K$ centers
at random from the observations. There exist more sophisticated ways
for  obtaining starting points such as the ROBIN algorithm, \citet{ROBIN}. This algorithm 
 looks for observations in high density areas that are far from each other. We use ROBIN algorithm once and  random starts $(H-1)$ times. Let 
$\boldsymbol{\mu }_{1}^{h},...,\boldsymbol{\mu }_{k}^{h},..,.\boldsymbol{\mu }
_{K}^{h}$, $1\leq h\leq H$ be the cluster centers obtained in the h$^{th}$ replication. 
Let
\begin{equation*}
h_{0}=\arg\min_{1\leq h\leq H}s_{\tau}(D(\mathbf{x}_{1},\boldsymbol{\mu}%
_{1}^{h},...,\boldsymbol{\mu}_{K}^{h_{{}}}),...,D(\mathbf{x}_{n},\boldsymbol{%
\mu }_{1}^{h},...,\boldsymbol{\mu}_{K}^{h}))
\end{equation*}
The final cluster centers are 
$(\hat{\boldsymbol{\mu}}_{1},...,\hat{\boldsymbol{\mu}}_{K}) = (\boldsymbol{\mu}_{1}^{h_{0}},...,\boldsymbol{\mu}_{K}^{h_{0}})$.

In our implementation we use $%
\rho _{1}=\rho (t/c_{1})$ and $\rho _{2}=\rho (t/c_{2})$, with $\rho $ equal
to the smooth hard-rejection loss function given by 
\begin{equation}
\rho (t)\!=\!\left\{ 
\begin{array}{lcc}
1.38t^{2} & \mbox{ if } & |t|<\frac{2}{3} \\ 
0.55\!-\!2.69t^{2}\!+\!10.76t^{4}\!-\!11.66t^{6}+\!4.04t^{8} & \mbox{ if } & 
\frac{2}{3}\leq |t|\leq 1 \\ 
1 & \mbox{ if} & |t|>1.%
\end{array}%
\right.  \label{rhoc}
\end{equation}%
The constants $c_{1}$ and $c_{2}$ depend on $p$ are such that the
tau scale is robust and efficient, following \citet{maronna2017robust}. The $%
\rho _{i}$ functions are smooth, quadratic at the center and quickly reach their
maximum value of one at $c_{i}$.
\medskip 

\noindent {\bf Outliers detection:} \label{OutlierRecognition} 
We flag an observation as a potential outlier if it falls outside the region 
$\mathcal{E}= \cup_{k=1}^K \mathcal{E}_k$, where $\mathcal{E}_k$ is the
confidence ellipsoid determined by 
\begin{equation}
\mathcal{E}_k=\{ \mathbf{x} \in\mathbb{R}^{p}: d^2(\mathbf{x},\hat{%
\boldsymbol{\mu}}_{k},\hat{\Sigma}_{k}) \leq\chi^{2}_{p,1-\beta} \},
\label{ellipse1}
\end{equation}
where $d^2(\mathbf{x}, \boldsymbol{\mu},\Sigma)= (\mathbf{x} -\boldsymbol{%
\mu })^{t}\Sigma^{-1}(\mathbf{x} -\boldsymbol{\mu })$ is the squared
Mahalanobis distance and $(\hat{\boldsymbol{\mu}}_{k},\hat{\Sigma}_{k})$ are
robust estimators of location and scatter matrix. We use the generalized S-estimators defined in   
 \citet{GSEPaper} and implemented in  \citet{PACKAGEGSE}. 

\subsection{Improved K-Tau.}

\label{improvedKtau} We
add the following  step to 
account for possibly  different cluster shapes and sizes. First, for each cluster $C_{k}$, we compute new S-estimators estimators of location and scatter
denoted  $\widetilde{\boldsymbol{\mu }}_k$
 and $\widetilde{\Sigma }_k$, 
 respectively, using  the 
package GSE  cited before in Section \ref{OutlierRecognition}. 
Improved  clusters are defined
as
\begin{equation*}
G_{k}=\left\{ \mathbf{x}_{i}:
\min_{1\leq j\leq K}d(\mathbf{x}_{i},%
\widetilde{\boldsymbol{\mu }}_{j},\widetilde{\Sigma }_{j}) 
 = d(\mathbf{x}_{i},\widetilde{\boldsymbol{\mu }}_{k},
\widetilde{\Sigma }_{k})\right\},
\end{equation*}
where 
$d^{2}(\mathbf{x},\boldsymbol{\mu },\Sigma )$ is 
	 the square  Mahalanobis distance
	$(\mathbf{x}-\boldsymbol{\mu })^T\Sigma
^{-1}(\mathbf{x}-\boldsymbol{\mu })$.
Once the new cluster are computed, possible
outliers are flagged using the procedure described in Section \ref{OutlierRecognition}. 
\medskip

\noindent{\bf Example:}
We use the  dataset  M5-data from \citet{garcia2008general}. These data consists of 1800 points generated from three
bivariate normals  with means and covariance matrices given by
\begin{align*}
\boldsymbol{\mu }_{1}& =(0,8),\text{ }\boldsymbol{\mu }_{2}=(8,0),%
\boldsymbol{\mu }_{3}=(-8,8), \mbox{~~~~ and} \\
\\
\Sigma _{1}& =\left( 
\begin{array}{cc}
1 & 0 \\ 
0 & 1%
\end{array}%
\right) ,\Sigma _{2}=\left( 
\begin{array}{cc}
45 & 0 \\ 
0 & 30%
\end{array}%
\right) ,\Sigma _{3}=\left( 
\begin{array}{cc}
15 & -10 \\ 
-10 & 15%
\end{array}%
\right) .
\end{align*}%
There are also 200 outliers generated from an
uniform distribution in a rectangular box around the bulk of data. The regular points are 20\%
in  cluster 1, 40 \% in  cluster 2 and 40\% in cluster
3. Panel A in  Figure \ref{figM5data} shows  the data with 0.95-level ellipsoid
around each cluster. Two of the clusters show some 
overlap. Panel B  shows the true cluster in colors red,
green and blue. The outliers are shown in black . 
The results from  K-Tau and improved K-Tau are shown in   panel C and   D, respectively. 
Improved K-Tau performs better because the shape of the
clusters is far from  spherical.

\begin{figure}[tbp]
\begin{equation*}
\begin{array}{ccc}
\includegraphics[scale=0.45]{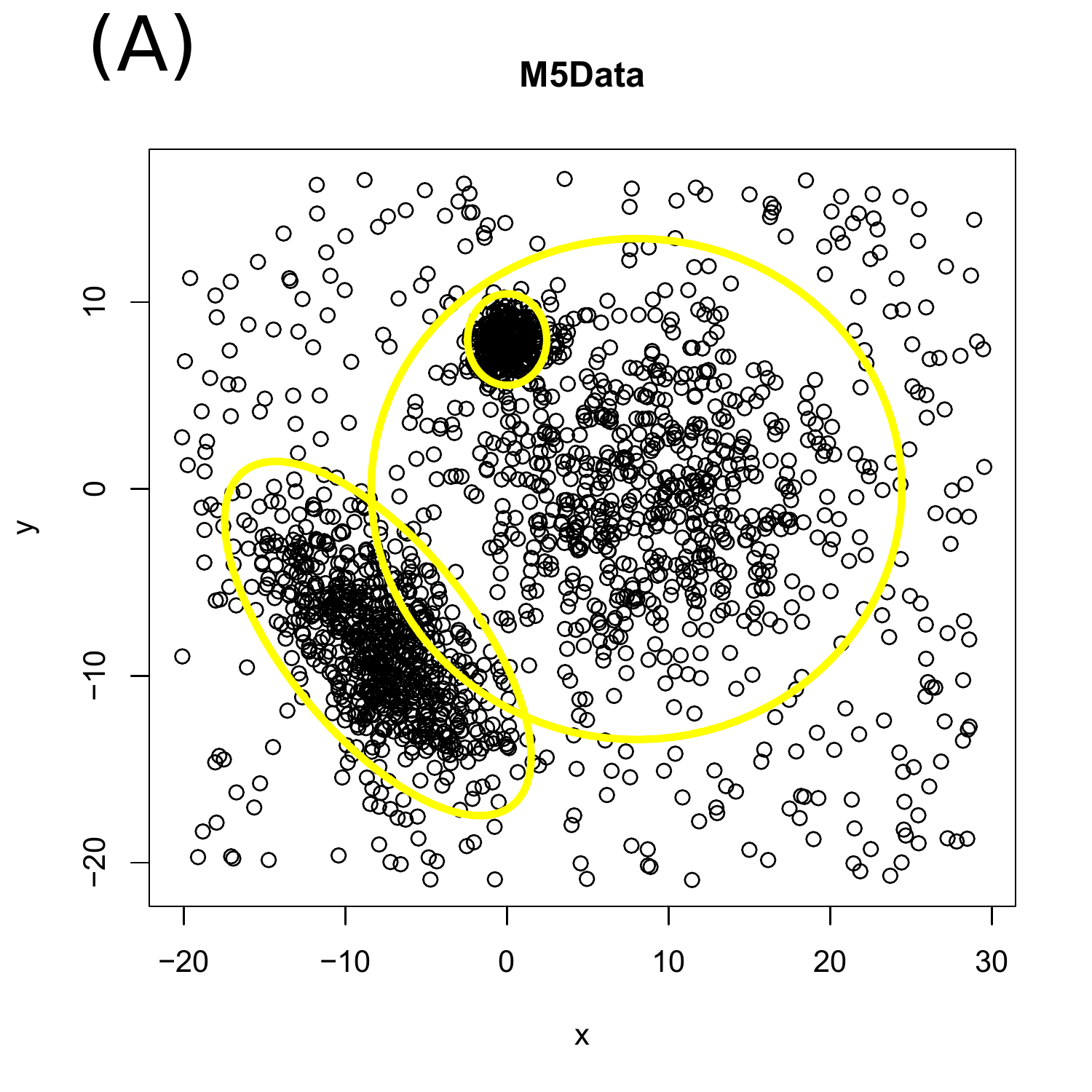} & %
\includegraphics[scale=0.45]{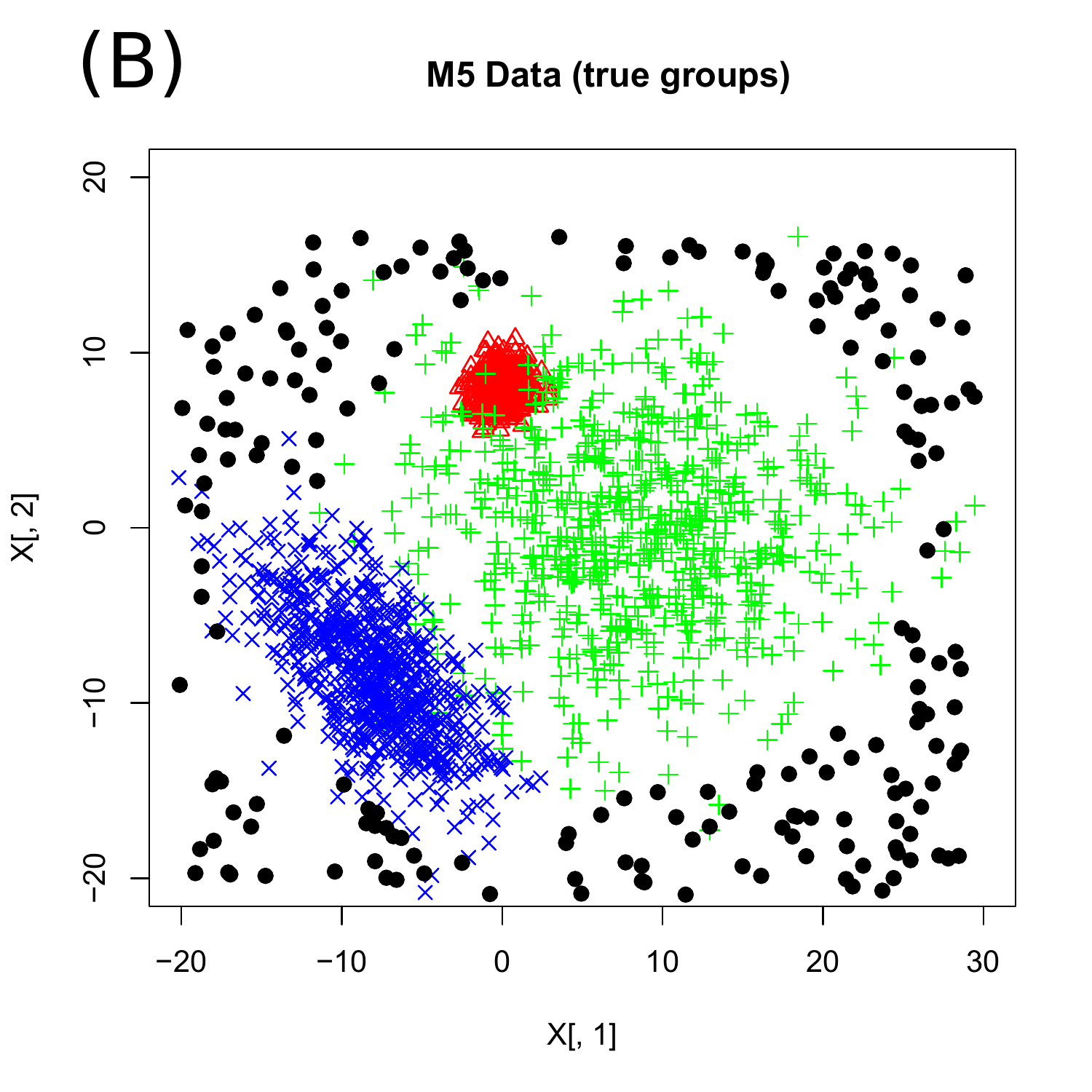} &  \\ 
\includegraphics[scale=0.45]{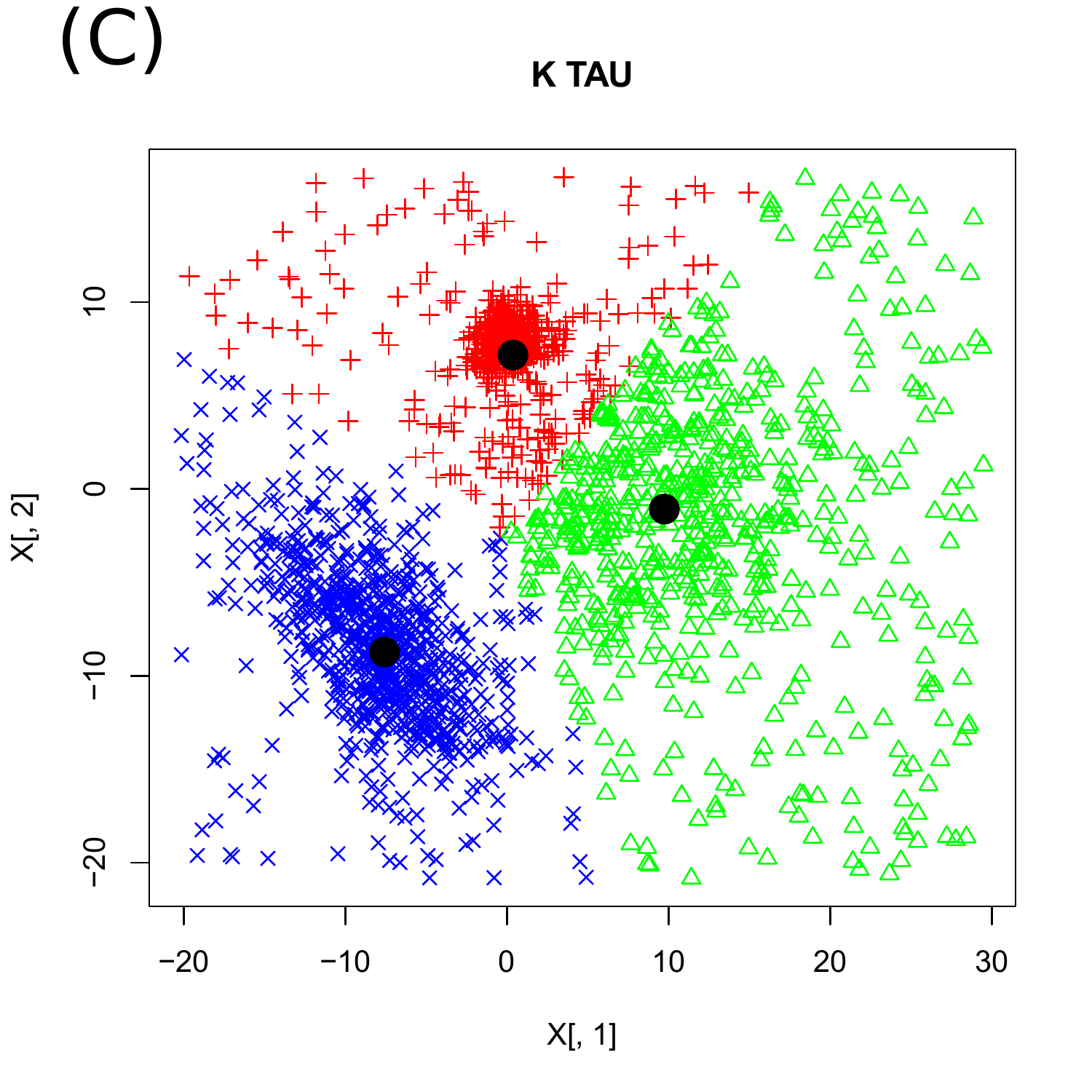} & %
\includegraphics[scale=0.45]{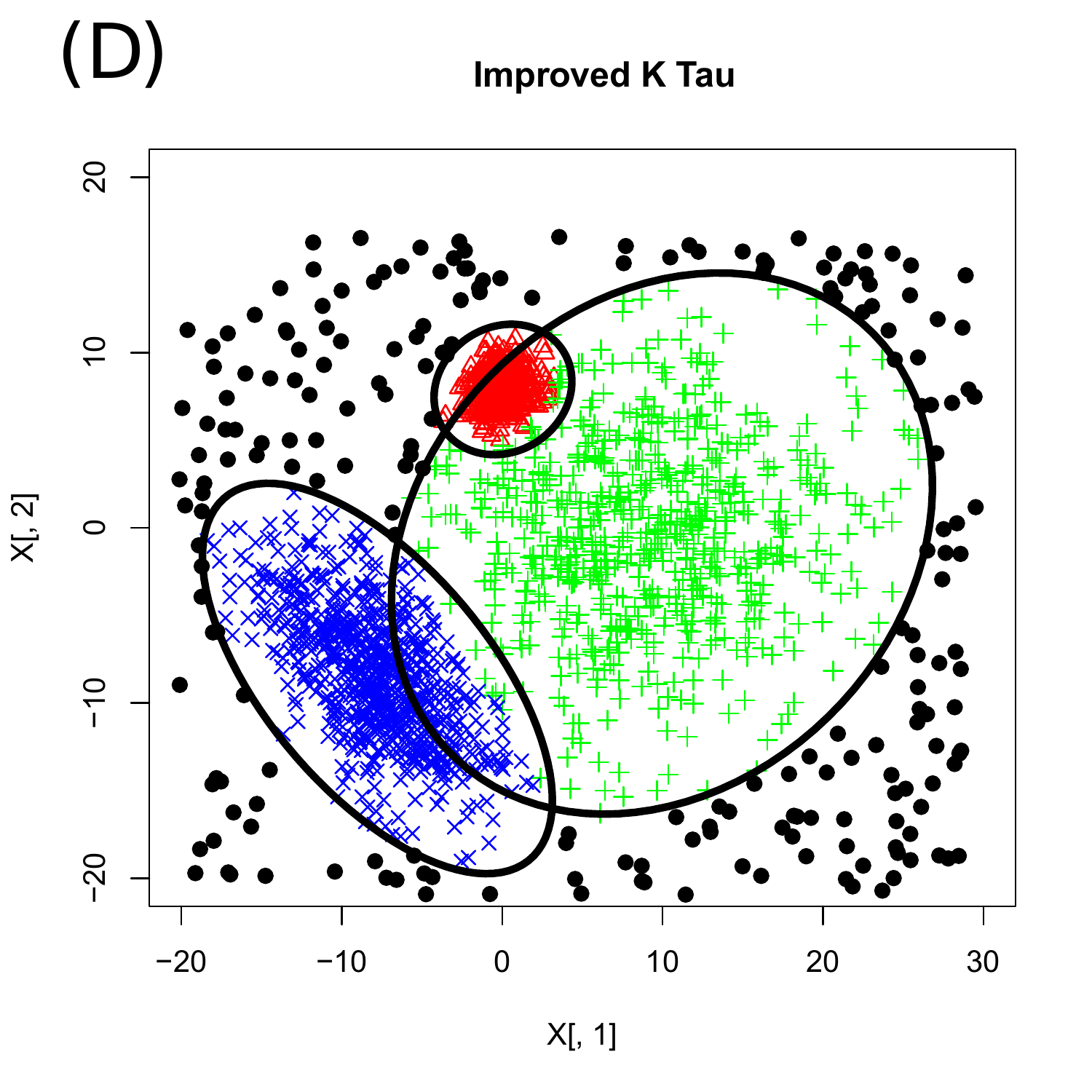} & 
\end{array}%
\end{equation*}%
\caption{(A) Data-set ``M5''.  Yellow solid lines show the $95\%$ confidence ellipsoids
 for the three clusters. (B) True clusters, with outliers shown in black. (C)
Results from K-Tau (D) Results from  improved K-Tau. }
\label{figM5data}
\end{figure}

\section{Consistency of K-Tau Centers}

Let $\mathbf{x}$ be a random vector in $\mathbb{R}^p$ with distribution
function $F$ and let $\mathcal{A}=\{\boldsymbol{\mu}_{1},...,\boldsymbol{\mu}%
_{K}\}$ be a set of $K$ distinct points in $R^p$. We define the $M$-scale
functional , $S_M(F)=S_M(F,\mathcal{A})$, as the solution in $s$ of the
equation 
\begin{equation*}
E_F\left(\rho_1\left(D(\mathbf{x},\mathcal{A})/s\right) \right) = 1/2,
\end{equation*}
with $D(\mathbf{x},\mathcal{A})$ defined by 
\begin{equation*}
D(\mathbf{x},\mathcal{A})=\min_{\boldsymbol{\mu} \mathcal{A} \in } \|\mathbf{%
x}-\boldsymbol{\mu}\|
\end{equation*}
The corresponding $\tau$-scale functional is 
\begin{equation*}
S_\tau(F,\mathcal{A}) = S_M(F,\mathcal{A})\left[E\left(\rho_1\left(D(\mathbf{%
x},\mathcal{A})/S_M(F,\mathcal{A})\right) \right) \right]^{1/2}.
\end{equation*}

%

The consistency of K-TAU centers is established the following theorem.

\begin{theorem}
\label{consistencia} Let $\mathbf{x}_{1},...,\mathbf{x}_{n}$ be a random
sample from $F$ and let $\mathcal{A}_n$ be the corresponding set of
K-TAU-cluster centers. Suppose that

A1. There exists $a<0.5$ such that the probability of any set of at most $k$
points doesn't exceed $a$.

A2. There exists a unique $\mathcal{A}_0$ such that $\min_{\mathcal{A}}
S_\tau(F,\mathcal{A}) = S_\tau(F,\mathcal{A}_0). \label{defpob} $

Then $\mathcal{A}_{n}\rightarrow ^{H}\mathcal{A}_{0}$ a.s., where .$%
\rightarrow ^{H}$ denotes convergence in the Haudorff metric.
\end{theorem}

%

Formal definition of Hausdorff metric can be found for example in reference 
\citet{Munkres}, colloquially let $\mathcal{A}, \mathcal{B}$ be subsets of $%
\mathbb{R}^p$, if Hausdorff distance between $\mathcal{A}$ and $\mathcal{B}$
is less than $\delta$ then for each $\mathbf{a} \in \mathcal{A}$, there
exists $\mathbf{b}=\mathbf{b}(\mathbf{a}) \in \mathcal{B}$ satisfying $\| 
\mathbf{a} -\mathbf{b} \| < \delta$ and reciprocally for each $\mathbf{b}
\in \mathcal{B}$ there exists $\mathbf{a}= \mathbf{a}(\mathbf{b}) \in 
\mathcal{A}$ such that $\|\mathbf{b} - \mathbf{a} \| < \delta$. Theorem \ref%
{consistencia} is proved in the Appendix II. 

\section{Simulation study}

We conducted a simulation study to compare improved K-Tau with K means and
 trimmed K means. More precisely, the compared procedures are: 
\bigskip

\noindent  \textbf{K means}. The classical K means introduced by \citet%
{hartigan1979algorithm}, using the default parameters in the function 
\texttt{kmeans} \color{black} from the \texttt{R} package \texttt{stats}.
\bigskip 

\noindent  \textbf{TK means}. The robust TK means introduced by \citet%
{EspaniolesTKMeans}, using the function \texttt{tkmeans} from the \texttt{R%
} package \texttt{tclust} by \citet{trimclustPackage}. The trimming constant $%
\alpha $ is set to $0.05,0.1,0.2\text{ and }0.3$ and the remaining tuning
parameters are set to their default values. Outliers are flagged using the mechanism
presented in Section \ref{OutlierRecognition}  with $\beta =0.01$. The confidence ellipsoid $%
\mathcal{E}_{k}$ has center and scatter matrix given by the classical mean
and covariance matrix of the points assigned to the k$^{th}$ cluster.
\bigskip 

\noindent  \textbf{IK-Tau}.The procedure proposed in Section 2.5 using
the function {\texttt{improvedktaucenters}} in the \texttt{R} package \texttt{%
ktaucenters}.  Outliers are flagged using the mechanism
presented in Section \ref{OutlierRecognition} with $\beta =0.01$. The
confidence ellipsoid $\mathcal{E}_{k}$ were obtained by applying the
function \texttt{GSE()} from the \texttt{R} package GSE \citet%
{PACKAGEGSE} to each group separately. 

\subsection*{Models used in the simulation}

Several scenarios, dimension $p=3, 5, 7, 10$ and total
number of clusters $K=3, 5, 7, 10$ are considered. Each cluster has
size $n_{k}=\theta \min \{p,4\}$, where $\theta $ 
can take the value 25 or 50 with equal probability.

The observations (for each cluster and replication) are generated from
multivariate normal distributions with mean 
\begin{equation*}
\boldsymbol{\mu}_k= \left\{ 
\begin{array}{crr}
20(-\frac{K}{2}+k)(1,1,1,\dots,1 ) & & %
\mbox{ if K is even} \\ 
&  &  \\ 
20(-\frac{K -1}{2}+k)(1,1,1,\dots,1) &  & %
\mbox{if  K is  odd}, \\ 
&  & 
\end{array}
\right.
\end{equation*}
 where  $k =1,2,\dots,K$, and covariance matrix $\Sigma_k$ generated as follows. First we generate a $%
p\times p$ matrix $U$ with $U_{ij} \sim \mathcal{U}(-1,1)$ and set ${U}{U}^t=%
{P} {\Lambda}{P}^t$, where $P$ is orthogonal. Second we create a diagonal
matrix $D$ with $D_{ii} \sim \mathcal{U}(1,10)$. Finally we set $\Sigma_{k}
= {P}{D}{P}^t.$

A proportion $\alpha=0.05$ of outliers are added to the sample. The outliers
are generated from a uniform distribution on a region obtained as follows.
We first expand by a factor of two the smallest box that contains the clean
data and then remove the points falling inside the  $99\%$ probability ellipsoids of the distributions
used to generate the clusters.

\subsection*{Performance measure}

Suppose that a clustering procedure is performed on a set of $n$
observations with known cluster membership. To evaluate the performance of the clustering procedure, for
each pair of observations $(\mathbf{x}_i, \mathbf{x}_j)$ with $1\leq i<j
\leq n$ we set $I(i,j)=1$ if observations $\mathbf{x}_i$ and $\mathbf{x}_j$
are both together or apart in the two  partitions. Otherwise, we
set $I(i,j)=0$. The \textit{Classification Error Rate} (CER) is then defined
by 
\begin{equation*}
\text{CER}=1 - 2\sum_{i<j}I(i,j)/[n(n-1)],
\end{equation*}
which is equal to one minus the Rand index
{proposed by \citet{CERREF}.}\color{black}
\noindent Robust clustering procedures are often used to flag outliers. Hence, in our
simulation study we collect the flagged observations in an extra cluster and
report the CER performance on the resulting $K+1$ clusters. %
%
%
%
%

\begin{figure}[!tbp]
\centering
\includegraphics[scale=0.6]{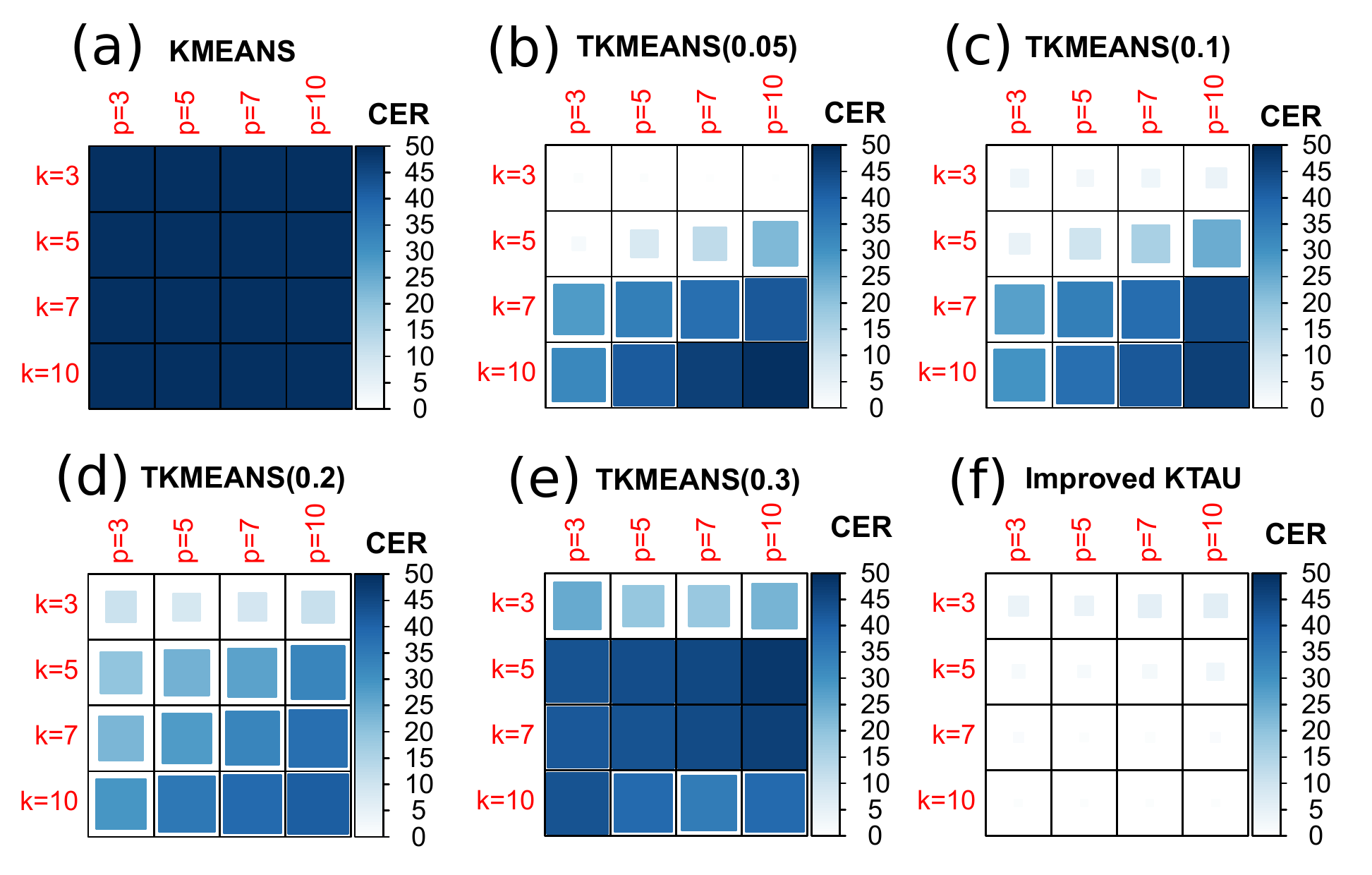}
\caption{Average Values of CER (multiplied by 1000) from different
scenarios. A lighter cell color indicates a better performance. }
\label{fig:resultadosSimulacionCER}
\end{figure}

\subsection*{Simulation results}

The results of our extensive simulation study are concisely presented on
Figure \ref{fig:resultadosSimulacionCER}. In each cell a color indicates the
average CER value for  the scenario identified by the number of clusters $k$ and the
dimension $p$. A darker color corresponds to a higher CER value.
The color-bar at the right indicates the CER-range values.

 
  Panel (a) in Figure \ref{fig:resultadosSimulacionCER} shows
that, as expected, K means cannot handle the outliers in the data and gives
overall very poor results. Panels (b)-(e) show the results for TK means for
different values of the trimming constant $\alpha =0.05,0.10,0.20\text{ and }%
0.30$. It is clear from these pictures that the best performance of TK means
correspond to the \textquotedblleft oracle trimming level\textquotedblright\  for this method ( $\alpha =0.05$).
Table \ref{1000*CER} contains the complete numerical results.

\begin{table}[ht]
\centering
\begin{tabular}{rrrrrrrr}
  \hline
$K$ & $p$ & K means & \multicolumn{4}{c}{TK Means} & IK-Tau \\ 
 & &  & $\alpha = 0.05$ & $\alpha = 0.10$ & $\alpha = 0.20$ & $\alpha = 0.30$ & \\
  \hline
  3 & 3 & 224.2 & 0.7 & 3.3 & 10.5 & 25.4 & 4.2 \\ 
  3 & 5 & 231.0 & 0.5 & 2.9 & 8.5 & 19.2 & 4.0 \\ 
  3 & 7 & 228.0 & 0.4 & 3.3 & 9.1 & 18.8 & 5.6 \\ 
  3 & 10 & 228.4 & 0.3 & 4.4 & 11.4 & 23.2 & 6.1 \\ 
  5 & 3 & 99.8 & 1.9 & 4.5 & 19.7 & 43.2 & 1.9 \\ 
  5 & 5 & 97.2 & 8.2 & 10.3 & 23.6 & 44.2 & 1.6 \\ 
  5 & 7 & 97.6 & 12.5 & 16.2 & 26.9 & 45.3 & 2.0 \\ 
  5 & 10 & 100.9 & 22.3 & 24.7 & 32.5 & 48.2 & 3.1 \\ 
  7 & 3 & 70.6 & 28.1 & 27.1 & 22.7 & 42.3 & 1.1 \\ 
  7 & 5 & 69.5 & 34.4 & 34.0 & 28.2 & 43.4 & 0.8 \\ 
  7 & 7 & 68.9 & 37.7 & 38.1 & 32.8 & 44.8 & 0.9 \\ 
  7 & 10 & 73.4 & 42.3 & 44.7 & 37.8 & 46.6 & 1.1 \\ 
  10 & 3 & 52.4 & 32.1 & 29.6 & 29.0 & 43.1 & 0.6 \\ 
  10 & 5 & 50.9 & 41.9 & 37.6 & 35.8 & 38.5 & 0.6 \\ 
  10 & 7 & 53.0 & 46.5 & 42.1 & 38.6 & 34.9 & 0.6 \\ 
  10 & 10 & 52.4 & 50.9 & 46.7 & 41.2 & 38.9 & 0.6 \\ 
   \hline
\end{tabular}
\caption{Simulation results (1000$\times$CER) for the considered clustering procedures.}
\label{1000*CER}
\end{table}


\section{Application}

\begin{figure}[!tbp]
\centering
\includegraphics[scale=0.40]{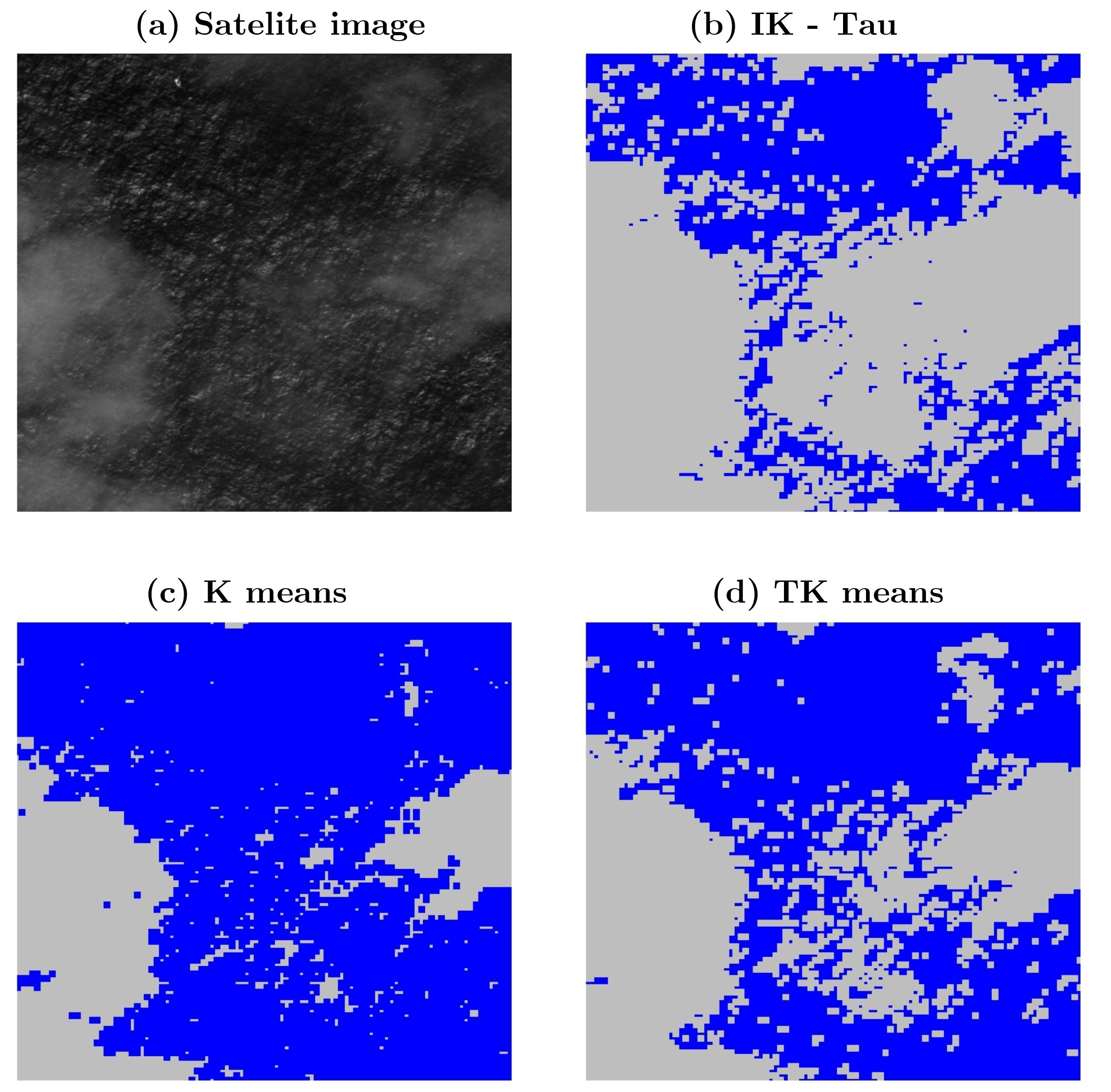}
\caption{Panel (a) gives the riginal Image. Panels (b), (c) and (d) show the image segmentation 
 produced by K means, TK means and IK-Tau.
The 10000 square cells are colored blue or gray according to their assignment to the
water or the clouds  clusters, respectively.}
\label{fig:cuatroJuntosAguaNube}
\end{figure}

\begin{figure}[!tbp]
\centering
\includegraphics[scale=0.5]{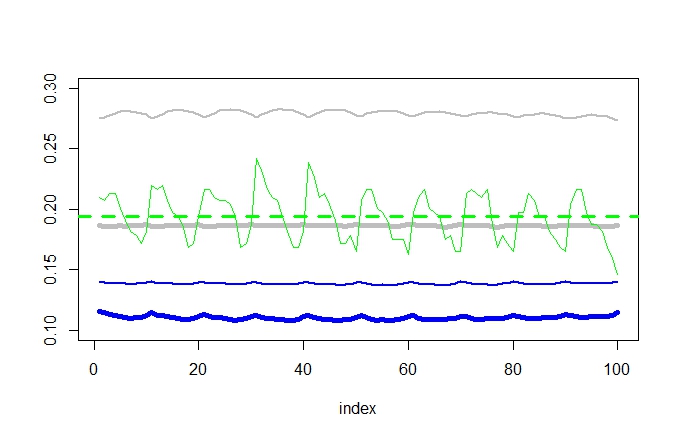}
\caption{Output Clusters Centers for IK-Tau (Wide line) and K means (Narrow
Line). Colors blue and gray correspond to the water and the cloud cluster
centers, respectively. A ramdomly selected observation from a low cloud part of the image is shown in  green. The  constant  green line gives the median value for this observation.}
\label{fig:cuatroJuntosAguaNubeCentros}
\end{figure}

\subsection*{Application 1: Cluster analysis of a satellite image}

Automatic unsupervised segmentation of satellite images is an important
problem in computer vision and automatic anomaly detection.
We analyze a satellite image 
covering $500m^{2}$ of the ocean (image provided by INFOSAT).  Each pixel conveys a gray-level intensity scaled between zero and one. Naturally, the image mostly consists  of two components: clouds and water. For the analysis, the high resolution image ($1$
pixel $=0.02m^{2}$) is divided into 10000  cells, each packing $10\times 10$ pixels. Hence, our dataset consists of 10000 points in the
one--hundred dimensional space $[0,1]^{100}$. Our goal is to
segment the image into two clusters (the \textit{cloud--cluster} and the 
\textit{water--cluster}) using IK-Tau. For comparison purposes we also apply 
 K Means and TK means. Figure \ref%
{fig:cuatroJuntosAguaNube} shows the original image and some clustering
results. Blue--colored cells correspond to water and gray--colored cells
correspond to clouds. 

The high altitude clouds (the brightest areas in the image) are
well recognized by all the considered procedures. On the other hand, due to
its lack of robustness, K means has difficulty segmenting the low--clouds
areas which bear relatively low gray-level intensities. This problem is
mainly caused by the presence of a patch of very high altitude clouds
with a very high gray--level intensity level. These outliers brings up
the intensity level of the K means clouds--cluster center. Figure \ref%
{fig:cuatroJuntosAguaNubeCentros} plots the index versus the gray--level
intensity for the clusters centers of K means, K-TAU, and a randomly chosen
low--cloud observation. The thick lines correspond to K-TAU, and the thin
lines correspond to K means. Low intensity clouds (a randomly chosen low
cloud is depicted by the green thick line in Figure \ref%
{fig:cuatroJuntosAguaNubeCentros} lie closer to the K means water--cluster
center and get mistakenly assigned to the water--cluster. On the other hand,
the considered robust method are not affected by the outliers and are
capable to correctly segment the very low clouds.

\subsection*{Application 2: Cluster analysis of high resolution picture}

\begin{figure}[!tbp]
\centering
\includegraphics[scale=0.60]{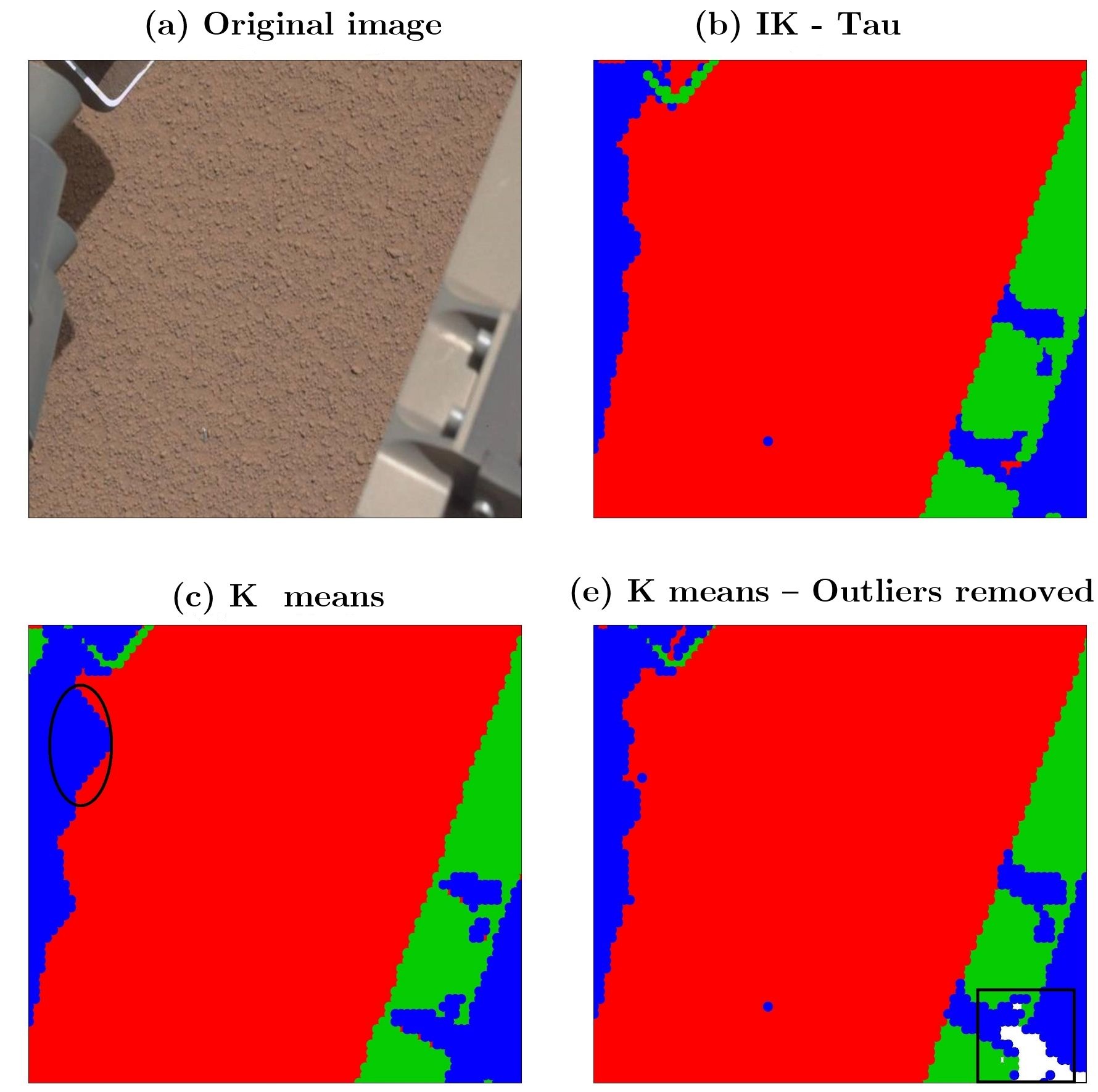}
\caption{ Original Image (panel (a)) and segmentation results obtained from IK-Tau, K means and Kmeans after removing outliers  (panels (b), (c) and (d)).  
Square cells are colored  red, green or blue according to their assignment
to the SND, SHM or OPM clusters. The black circle in panel (c) pinpoints the sand shadow region in the original image. The blank area in panel (d) are to the removed outlying cells.}
\label{fig:cuatroJuntosTornillo}
\end{figure}

\begin{figure}[!tbp]
\centering
\includegraphics[scale=0.4]{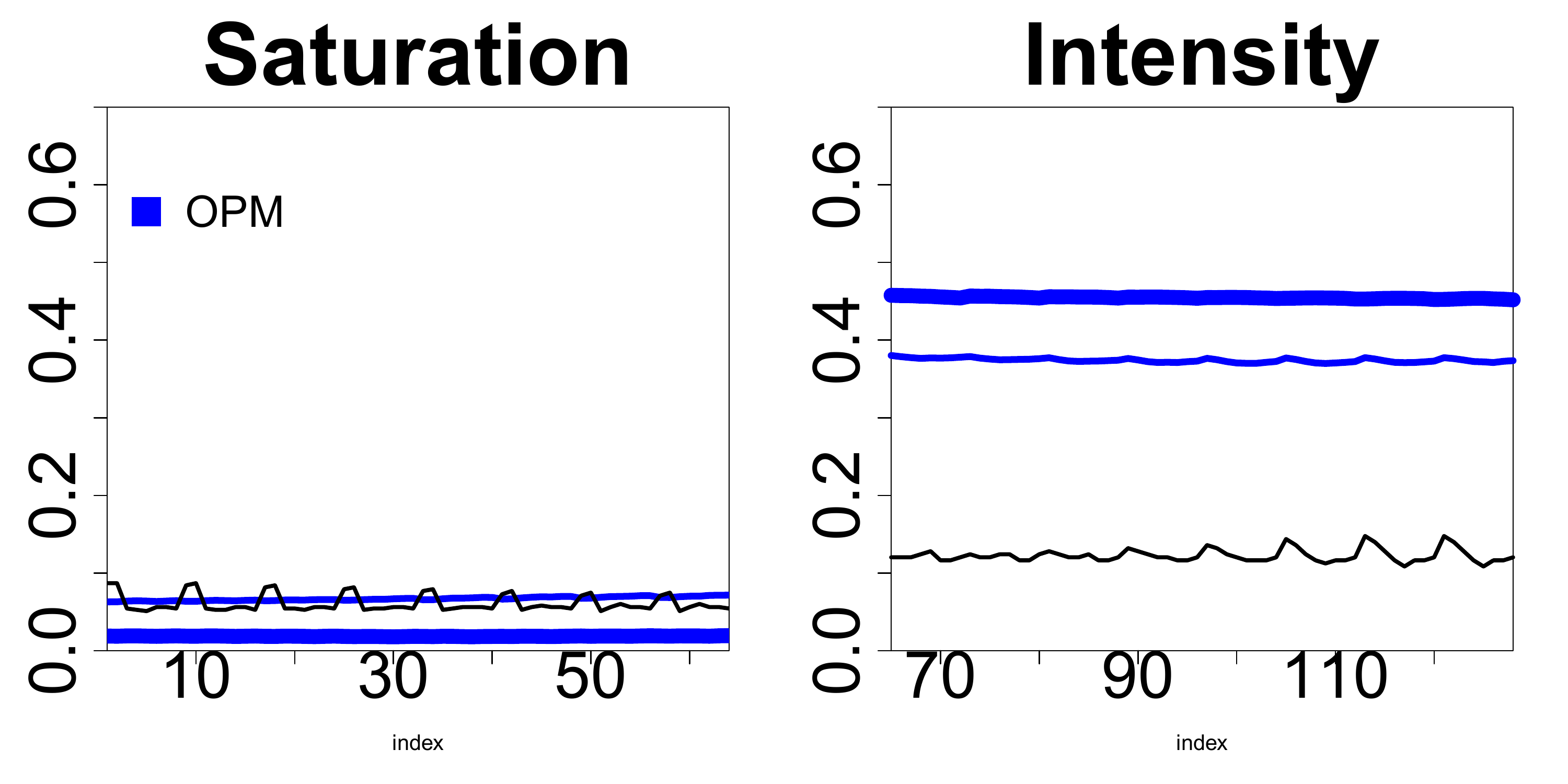} 

\caption{The Opaque Metal cluster Centers corresponding from the improved K-TAU (Wide line) and K means (Narrow Line).  An outlying observation is depicted by the thin
black line.}
\label{fig:centrosRobustosTornillo} 
\end{figure}

In this application we consider is a high resolution colored picture of $495
\times 664$ pixels covering an area of $4.0\times 7.5$ $cm^2$, provided
by NASA \citet{NASA} and displayed in  Figure \ref{fig:cuatroJuntosTornillo} (a). This
image was taken by the NASA's Mars rover Curiosity, and shows the sand soil and metal from the  Mars rover itself. Each pixel has assigned three numbers representing the intensity levels of the R, G and B channels, scaled between $0$ and $1$. For example $(0,0,0)$ represents black, $(1,1,1)$ white, $%
(0,1,0)$ green, etc. Because the $(R,G,B)$ variables are usually highly
correlated, a common practice in image segmentation is to transform
the $(R,G,B)$ values into the \textit{saturation} (S) and \textit{%
intensity} (I) values, where  $I= (R + G + B)/3$ and $S = 1 - \min(R, G,
B)/I$ with $S = 0 $ when $I=0$. See Chen et al, 2001 \cite{HSI}, for
further details.

The pixels in the original image are arranged into  $8\times 8$ square cells. Since,  
each pixel has two numbers $(S,I)$, each cell (observation)  represents a point in  in the one
hundred and twenty eight dimensional space $[0,1]^{128}$. Our objective is to segment the
image into three clusters, namely the \textit{shinning metal} (SHM),
the \textit{opaque metal} (OPM) and the \textit{sand} (SND) clusters. As
in the previous application, we use the three clustering algorithms:  K means,
TK means and IK-Tau  with $K=3$. Since the two robust procedures give
similar results, only those of IK-Tau are displayed in Figure  
\ref{fig:cuatroJuntosTornillo}. The nonrobust K means
is affected by the presence a small fraction of very dark metal cells in the
 lower right corner of Figure \ref{fig:cuatroJuntosTornillo} (a). As shown
in Figure \ref{fig:centrosRobustosTornillo}, the dark metal cells have
very low I level, and very high S level. These outliers bring
up the S level and down the I level of the  OPM cluster center in the case of K means. We notice
that the outliers represent $1.5\%$ of the image and $15\%$ of the
OPM cluster. As a consequence the shaded sand region enclosed by the
ellipse in Figure \ref{fig:cuatroJuntosTornillo} (c) are incorrectly assigned  to the OPM cluster by K means. To
validate this reasoning we recompute the K means clusters after removing
the aforementioned outliers (the cells delimited by the rectangle in
the lower right corner of panel (d) of Figure \ref{fig:cuatroJuntosTornillo}). Now the K means results are consistent with those of the robust clustering
procedures.

\subsection*{ Searching for lost objects}

We  now show examples of  how large image data in conjunction with robust
cluster analysis could be used in computer-aided searches for lost objects. Our
first example continues from Application  1. In this case, the lost object is
the Tunante II, a 12.5 meter-long yacht with four crew on board lost during
a storm off the coast of Brazil. Note that the boat is made of a material
mostly absent in the satellite image. Our second example continues from
Application 2 and the lost object is a small metal screw detached from the
Curiosity Rover, during its exploration of Mars. In this case the lost object is made of a material (metal) well
represented in the image making its automatic finding more difficult.

\subsubsection*{Looking for the lost boat}

Naturally we hope that the small yacht reflects the signal differently from
the water and clouds in the image. Therefore, the cells containing the boat
should appear as outliers in the clustering results. The cell size in
Application 1 has been intendedly chosen so that the boat is fully contained by
at most four neighboring cells. 
Using the results from Application 1, we identify the most extreme outlier, that is,
the cell lying further away from its cluster center (see Figure \ref%
{fig:mainResult}). Proceeding in this way, all the considered robust and nonrobust cluster algorithms
 succeed in locating the yacht  Tunante. 

\subsubsection*{Searching for the lost screw}

In this example we search for a small screw detached form the land rove Mars, 
using the cluster results from Application  2.
The screw is made of a material -- metal -- that makes up $25\%$ of the
image. The image data consists of an $n\times p$ data matrix $X$ with $%
n=5063 $ rows ( each row corresponding to a cell of $8\times 8$ pixels) and $%
p=128$ columns (each column corresponding to the saturation and the
intensity values for the $64$ pixels in each cell). The clustering results
from Application 2 yielded three clusters of sizes $n_{1}=3918$, $n_{2}=617$ and $%
n_{3}=528$ corresponding to sand, opaque metal and shinning metal,
respectively. Assuming that we know the type of material of the missing
object (e.g. a screw made of opaque metal) we can restrict attention
to the $n_{2}\times 2$ geographic submatrix that gives the position of the
cells assigned to the OPM cluster (see Figure \ref{fig:geographic}). We perform a second robust cluster analysis on these geographic data. 
Clearly from Figure \ref{fig:geographic} (a) any isolated outliers from a
second robust cluster analysis of this geographic data are candidates for
the location of the missing screw. The robust analysis exposes  the
remarkable isolated point in Figure \ref{fig:geographic} (a) which indeed
corresponds to the missing screw. The non-robust analysis leads to the much
less informative Figure \ref{fig:geographic} (b).

\begin{figure}[!tbp]
\centering
\includegraphics[scale=0.3]{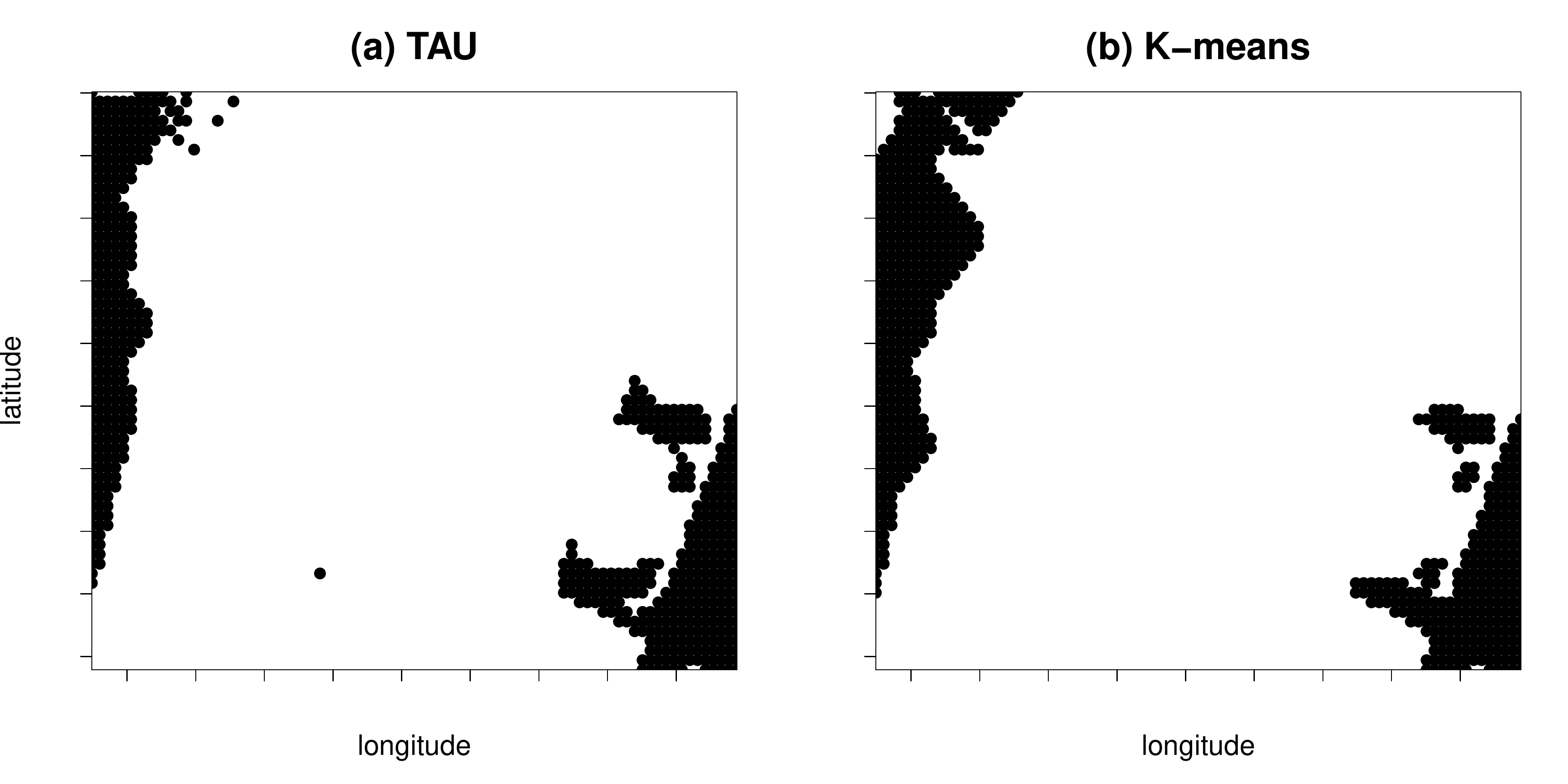}
\caption{geographic submatrix corresponding to Opaque Metal cluster obtained
by (a) K-TAU, (b) K means }
\label{fig:geographic}
\end{figure}

\begin{figure}[!tbp]
\includegraphics[width=0.485\textwidth]{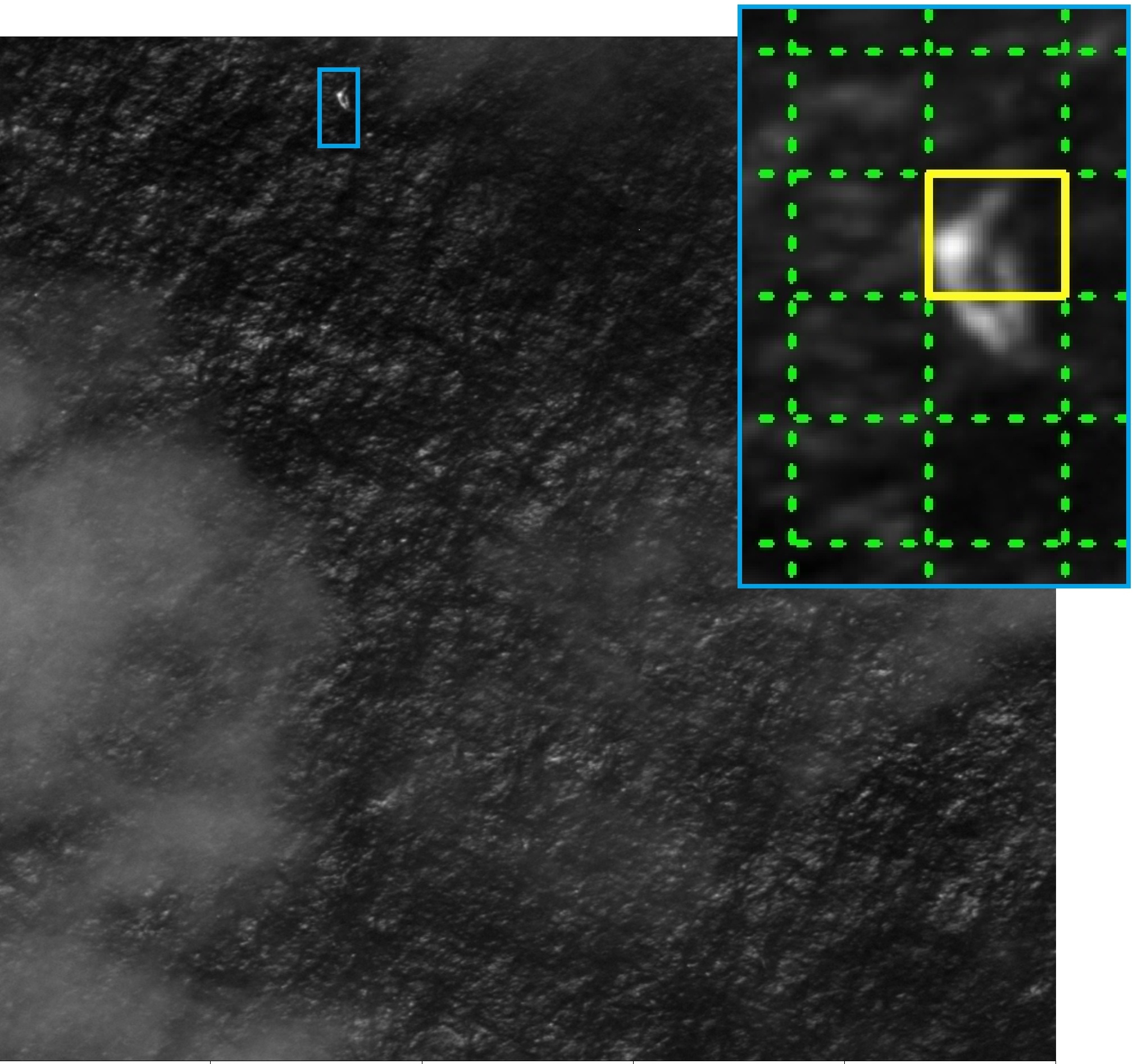} %
\includegraphics[width=0.47\textwidth]{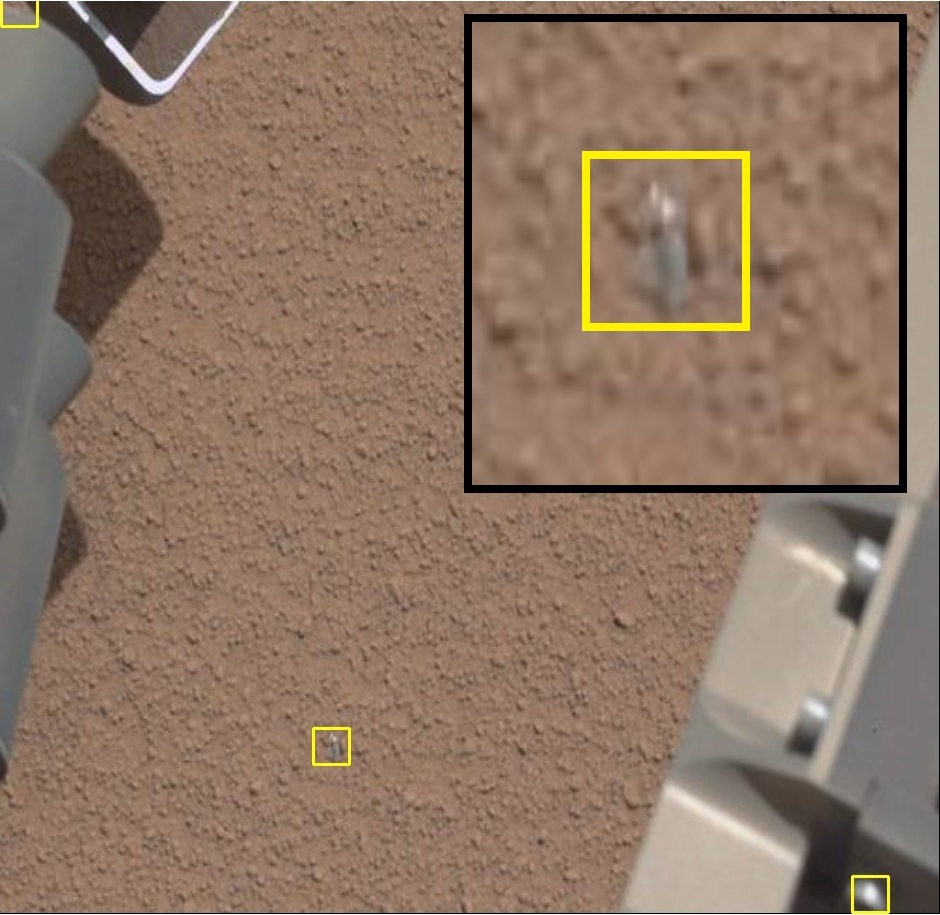}
\caption{\textbf{Left:} The boat found by the algorithm. Square cells of
size $10\times 10$ are shown in green dotted lines, the observation found is
in yellow. \textbf{Right:} The screw found by the 2-step procedure. The
candidates are indicated in yellow squares on the left upper corner, right
lower corner, and in the middle, the region of the screw is expanded in the
black rectangle.}
\label{fig:mainResult}
\end{figure}

\clearpage


\section*{Appendix I: Derivation of Estimating equations}

\label{EstimatingEq}

\subsection*{Notation}

We establish the notation as follows 
\begin{equation*}
\boldsymbol{\mu}= (\boldsymbol{\mu}_{1},\boldsymbol{\mu}_{2},\dots,%
\boldsymbol{\mu}_K), \mbox{where } \boldsymbol{\mu}_{k} \in \mathbb{R}^p,
\end{equation*}
\begin{equation*}
\boldsymbol{d}(\boldsymbol{\mu})=(d_1(\boldsymbol{\mu}),d_2(\boldsymbol{\mu}%
), \dots, d_n(\boldsymbol{\mu})), \quad \mbox{where } \quad d_i(\boldsymbol{%
\mu}) = \min_{1\leq k \leq K} {\|\mathbf{x}_i-\boldsymbol{\mu}_k\|}, \ \
1\leq i \leq n.
\end{equation*}
First of all, we set $k$, $1\leq k\leq K$, and compute the derivatives $%
\partial d_{i}(\boldsymbol{\mu })/\partial \boldsymbol{\mu }_{k}\ $ and $\
\partial s(\boldsymbol{d}(\boldsymbol{\mu }))/\partial \boldsymbol{\mu }%
_{k}\ $

\begin{itemize}
\item derivation of $\partial d_{i}(\boldsymbol{\mu })/\partial \boldsymbol{%
\mu }_{k}$ 
\begin{equation}
\frac{\partial d_{i}(\boldsymbol{\mu })}{\partial \boldsymbol{\mu }_{k}}%
=\left\{ 
\begin{array}{cc}
-\displaystyle\frac{(\mathbf{x}_{i}-\boldsymbol{\mu }_{k})}{||\mathbf{x}_{i}-%
\boldsymbol{\mu }_{k}||} & \quad \mbox{si }\ \mathbf{x}_{i}\in G_{k} \\ 
&  \\ 
0 & \quad \mbox{si }\ \mathbf{x}_{i}\notin G_{k}%
\end{array}%
\right.  \label{derivadadi}
\end{equation}

\item derivation of $\partial s(\boldsymbol{d}(\boldsymbol{\mu }))/\partial 
\boldsymbol{\mu }_{k}$

$s(\boldsymbol{d}(\boldsymbol{\mu }))$, satisfies (\ref{sdefff}), therefore, by
implicit differentiation, 
\begin{equation*}
\frac{\partial }{\partial \boldsymbol{\mu }_{k}}\left( \frac{1}{n}%
\sum_{i=1}^{n}\rho _{1}\left( \frac{d_{i}(\boldsymbol{\mu })}{s(\boldsymbol{d%
}(\boldsymbol{\mu }))}\right) \right) =0
\end{equation*}%
\begin{equation*}
\frac{1}{n}\sum_{i=1}^{n}\psi _{1}\left( \frac{d_{i}(\boldsymbol{\mu })}{s(%
\boldsymbol{d}(\boldsymbol{\mu }))}\right) \frac{\partial }{\partial 
\boldsymbol{\mu }_{k}}\left( \frac{d_{i}(\boldsymbol{\mu })}{s(\boldsymbol{d}%
(\boldsymbol{\mu }))}\right) =0,
\end{equation*}%
this implies 
\begin{equation*}
\sum_{i=1}^{n}\psi _{1}\left( \frac{d_{i}(\boldsymbol{\mu })}{s(\boldsymbol{d%
}(\boldsymbol{\mu }))}\right) \left( \frac{\partial d_{i}(\boldsymbol{\mu })%
}{\partial \boldsymbol{\mu }_{k}}s(\boldsymbol{d}(\boldsymbol{\mu }))-d_{i}(%
\boldsymbol{\mu })\frac{\partial s(\boldsymbol{d}(\boldsymbol{\mu }))}{%
\partial \boldsymbol{\mu }_{k}}\right) =0,
\end{equation*}%
From here, we obtain 
\begin{equation*}
\frac{\partial s(\boldsymbol{d}(\boldsymbol{\mu }))}{\partial \boldsymbol{%
\mu }_{k}}=-\frac{\sum_{i=1}^{n}\psi _{1}\left( \frac{d_{i}(\boldsymbol{\mu }%
)}{s(\boldsymbol{d}(\boldsymbol{\mu }))}\right) \frac{\partial d_{i}(%
\boldsymbol{\mu })}{\partial \boldsymbol{\mu }_{k}}}{\sum_{i=1}^{n}\psi
_{1}\left( \frac{d_{i}(\boldsymbol{\mu })}{s(\boldsymbol{d}(\boldsymbol{\mu }%
))}\right) \frac{d_{i}(\boldsymbol{\mu })}{s(\boldsymbol{d}(\boldsymbol{\mu }%
))}}
\end{equation*}%
Let $B(\boldsymbol{\mu })$ be a real number depending on $\boldsymbol{\mu }$ 
\begin{equation*}
B(\boldsymbol{\mu })=\frac{1}{n}\sum_{i=1}^{n}\psi _{1}\left( \frac{d_{i}(%
\boldsymbol{\mu })}{s(\boldsymbol{d}(\boldsymbol{\mu }))}\right) \frac{d_{i}(%
\boldsymbol{\mu })}{s(\boldsymbol{d}(\boldsymbol{\mu }))}
\end{equation*}%
thus, 
\begin{equation*}
\frac{\partial s(\boldsymbol{d}(\boldsymbol{\mu }))}{\partial \boldsymbol{%
\mu }_{k}}=-\frac{1}{B(\boldsymbol{\mu })}\sum_{i=1}^{n}\psi _{1}\left( 
\frac{d_{i}(\boldsymbol{\mu })}{s(\boldsymbol{d}(\boldsymbol{\mu }))}\right) 
\frac{\partial d_{i}(\boldsymbol{\mu })}{\partial \boldsymbol{\mu }_{k}},
\end{equation*}%
by using (\ref{derivadadi}), we get 
\begin{equation}
\frac{\partial s(\boldsymbol{d}(\boldsymbol{\mu }))}{\partial \boldsymbol{%
\mu }_{k}}=-\frac{1}{B(\boldsymbol{\mu })}\sum_{\mathbf{x}_{i}\in
G_{k}}^{n}\psi _{1}\left( \frac{\mathbf{x}_{i}-\boldsymbol{\mu }_{k}}{s(%
\boldsymbol{d}(\boldsymbol{\mu }))}\right) \frac{(\mathbf{x}_{i}-\boldsymbol{%
\mu }_{k})}{||\mathbf{x}_{i}-\boldsymbol{\mu }_{k}||}.  \label{derivadas1}
\end{equation}
\end{itemize}

Now we find the centers that minimize $J_{\tau }$ given by 
\begin{equation*}
J_{\tau }^{2}\left( \boldsymbol{\mu }_{1},...,\boldsymbol{\mu }_{K}\right)
=s^{2}(\boldsymbol{d}(\boldsymbol{\mu }))\frac{1}{n}\sum_{i=1}^{n}\rho
_{2}\left( \frac{d_{i}(\boldsymbol{\mu })}{s(\boldsymbol{d}(\boldsymbol{\mu }%
))}\right) ,
\end{equation*}

The estimating equations for the clusters centers are obtained by equating
the derivative of $J_{\tau }$ to zero, 
\begin{equation*}
\frac{\partial }{\partial {\boldsymbol{\mu }_{k}}}\left( J_{\tau }\left( 
\boldsymbol{\ \mu _{1}},...,\boldsymbol{\mu _{k}}\right) \right) =0,
\end{equation*}

\begin{eqnarray*}
\frac{\partial }{\partial {\boldsymbol{\mu }_{k}}}\left( J_{\tau }\left( 
\boldsymbol{\mu }_{1},...,\boldsymbol{\mu }_{k}\right) \right) &=&\frac{%
\partial }{\partial \boldsymbol{\mu }_{k}}\left( s^{2}(\boldsymbol{d}(%
\boldsymbol{\mu }))\frac{1}{n}\sum_{i=1}^{n}\rho _{2}\left( \frac{d_{i}(%
\boldsymbol{\mu })}{s(\boldsymbol{d}(\boldsymbol{\mu }))}\right) \right) \\
{} &=&\frac{\partial s^{2}(\boldsymbol{d}(\boldsymbol{\mu }))}{\partial 
\boldsymbol{\mu }_{k}}\frac{1}{n}\sum_{i=1}^{n}\rho _{2}\left( \frac{d_{i}(%
\boldsymbol{\mu })}{s(\boldsymbol{d}(\boldsymbol{\mu }))}\right) +s^{2}(%
\boldsymbol{d}(\boldsymbol{\mu }))\frac{\partial }{\partial \boldsymbol{\mu }%
_{k}}\left( \frac{1}{n}\sum_{i=1}^{n}\rho _{2}\left( \frac{d_{i}(\boldsymbol{%
\mu })}{s(\boldsymbol{d}(\boldsymbol{\mu }))}\right) \right) \\
{} &=&2s(\boldsymbol{d}(\boldsymbol{\mu })){\frac{\partial s(\boldsymbol{d}(%
\boldsymbol{\mu }))}{\partial \boldsymbol{\mu }_{k}}}\frac{1}{n}%
\sum_{i=1}^{n}\rho _{2}\left( \frac{d_{i}(\boldsymbol{\mu })}{s(\boldsymbol{d%
}(\boldsymbol{\mu }))}\right) + \\
{} &&+{s^{2}(\boldsymbol{d}(\boldsymbol{\mu }))}\frac{1}{n}%
\sum_{i=1}^{n}\psi _{2}\left( \frac{d_{i}(\boldsymbol{\mu })}{s(\boldsymbol{d%
}(\boldsymbol{\mu }))}\right) \frac{\frac{\partial d_{i}(\boldsymbol{\mu })}{%
\partial \boldsymbol{\mu }_{k}}s(\boldsymbol{d}(\boldsymbol{\mu }))-d_{i}(%
\boldsymbol{\mu })\frac{\partial s(\boldsymbol{d}(\boldsymbol{\mu }))}{%
\partial \boldsymbol{\mu }_{k}}}{{s^{2}(\boldsymbol{d}(\boldsymbol{\mu }))}},
\\
&&
\end{eqnarray*}%
grouping $s(\boldsymbol{d}(\boldsymbol{\mu }))\frac{\partial s(\boldsymbol{d}(%
\boldsymbol{\mu }))}{\partial \boldsymbol{\mu }_{k}}$ from the above
equation 
\begin{equation*}
0=s(\boldsymbol{d}(\boldsymbol{\mu }))\frac{\partial s(\boldsymbol{d}(%
\boldsymbol{\mu }))}{\partial \boldsymbol{\mu }_{k}}\left\{ \frac{1}{n}%
\sum_{i=1}^{n}2\rho _{2}\left( \frac{d_{i}(\boldsymbol{\mu })}{s(\boldsymbol{%
d}(\boldsymbol{\mu }))}\right) -\psi _{2}\left( \frac{d_{i}(\boldsymbol{\mu }%
)}{s(\boldsymbol{d}(\boldsymbol{\mu }))}\right) \frac{d_{i}(\boldsymbol{\mu }%
)}{s(\boldsymbol{d}(\boldsymbol{\mu }))}\right\} +
\end{equation*}%
\begin{equation*}
+\frac{1}{n}\sum_{i=1}^{n}\psi _{2}\left( \frac{d_{i}(\boldsymbol{\mu })}{s(%
\boldsymbol{d}(\boldsymbol{\mu }))}\right) \frac{\partial d_{i}(\boldsymbol{%
\mu })}{\partial \boldsymbol{\mu }_{k}}s(\boldsymbol{d}(\boldsymbol{\mu })),
\end{equation*}%
defining 
\begin{equation*}
A(\boldsymbol{\mu })=\frac{1}{n}\sum_{i=1}^{n}2\rho _{2}\left( \frac{d_{i}(%
\boldsymbol{\mu })}{s(\boldsymbol{d}(\boldsymbol{\mu }))}\right) -\psi
_{2}\left( \frac{d_{i}(\boldsymbol{\mu })}{s(\boldsymbol{d}(\boldsymbol{\mu }%
))}\right) \frac{d_{i}(\boldsymbol{\mu })}{s(\boldsymbol{d}(\boldsymbol{\mu }%
))},
\end{equation*}%
we get 
\begin{equation*}
0=s(\boldsymbol{d}(\boldsymbol{\mu }))A(\boldsymbol{\mu })\frac{\partial s(%
\boldsymbol{d}(\boldsymbol{\mu }))}{\partial \boldsymbol{\mu }_{k}}+\frac{1}{%
n}\sum_{i=1}^{n}\psi _{2}\left( \frac{d_{i}(\boldsymbol{\mu })}{s(%
\boldsymbol{d}(\boldsymbol{\mu }))}\right) s(\boldsymbol{d}(\boldsymbol{\mu }%
))\frac{\partial d_{i}(\boldsymbol{\mu })}{\partial \boldsymbol{\mu }_{k}}.
\end{equation*}%
By using equations (\ref{derivadadi}), (\ref{derivadas1}), and multiplying
by $(-1)$, 
\begin{equation*}
0=s(\boldsymbol{d}(\boldsymbol{\mu }))\frac{A(\boldsymbol{\mu })}{D(%
\boldsymbol{\mu })}\sum_{\mathbf{x}_{i}\in G_{k}}\psi _{1}\left( \frac{d_{i}(%
\boldsymbol{\mu })}{s(\boldsymbol{d}(\boldsymbol{\mu }))}\right) \frac{(%
\mathbf{x}_{i}-\boldsymbol{\mu }_{k})}{d_{i}(\boldsymbol{\mu })}+
\end{equation*}%
\begin{equation*}
+\frac{1}{n}\sum_{\mathbf{x}_{i}\in G_{k}}\psi _{2}\left( \frac{d_{i}(%
\boldsymbol{\mu })}{s(\boldsymbol{d}(\boldsymbol{\mu }))}\right) s(%
\boldsymbol{d}(\boldsymbol{\mu }))\frac{(\mathbf{x}_{i}-\boldsymbol{\mu }%
_{k})}{d_{i}(\boldsymbol{\mu })},
\end{equation*}%
rearranging, 
\begin{equation*}
0=\sum_{\mathbf{x}_{i}\in G_{k}}\left[ A(\boldsymbol{\mu })\psi _{1}\left( 
\frac{d_{i}(\boldsymbol{\mu })}{s(\boldsymbol{d}(\boldsymbol{\mu }))}\right) 
\frac{s(\boldsymbol{d}(\boldsymbol{\mu }))}{d_{i}(\boldsymbol{\mu })}+D(%
\boldsymbol{\mu })\psi _{2}\left( \frac{d_{i}(\boldsymbol{\mu })}{s(%
\boldsymbol{d}(\boldsymbol{\mu }))}\right) \frac{s(\boldsymbol{d}(%
\boldsymbol{\mu }))}{d_{i}(\boldsymbol{\mu })}\right] (\mathbf{x}_{i}-%
\boldsymbol{\mu }_{k}),
\end{equation*}%
finally, 
\begin{equation*}
0=\sum_{\mathbf{x}_{i}\in G_{k}}\left[ \frac{A(\boldsymbol{\mu })\psi
_{1}\left( t\right) +B(\boldsymbol{\mu })\psi _{2}\left( t\right) }{t}\right]
_{t=\frac{d_{i}(\boldsymbol{\mu })}{s(\boldsymbol{d}(\boldsymbol{\mu }))}}(%
\mathbf{x}_{i}-\boldsymbol{\mu }_{k})
\end{equation*}%
From here (\ref{taueq1}) can be derived easily. \clearpage

\section*{Appendix II: Consistency}
\label{consistency}
Strong consistency for the classic K means was given by \citet{PollardPaper},
while strong consistency for the robust TK means was given by \citet%
{EspaniolesTKMeans}. Both works prove convergence by showing that the
Hausdorff distance ($d_H$) between the true centers an the estimated ones
tends to zero as the sample size tends to infinity. We provide an analogous
result for K-TAU. Our proof uses the results from Lemmas stated and
proved in the following section.

\subsection*{Consistency- Lemmas}

\begin{lemma}
\label{funcCompactas} Let $\varphi: \boldsymbol{\mathcal{K}_0} \to \mathbb{R}
$ be a continuous function with a unique minimum $\nu_0= \varphi(\mathcal{A}%
_0)$, where $\boldsymbol{\mathcal{K}}_0=\{ \mathcal{A}\subseteq \overline{B}%
, \#\mathcal{A} \leq K\}$, and $\overline{B}$ is a closed ball  in $\mathbb{R}^p$. Suppose that $(%
\mathcal{A}_n)_{n\in \mathbb{N}} \subseteq \boldsymbol{\mathcal{K}_0}$ and
assume that the following property is satisfied 
\begin{equation*}
\forall \eta>0 \ \exists n_0: \varphi(\mathcal{A}_n) < \eta + \nu_0 \ \ \
\forall n\geq n_0,
\end{equation*}
then $d_H(\mathcal{A}_n, \mathcal{A}_0) \to 0$.
\end{lemma}

\begin{proof}
First of all, we note that $(\boldsymbol{\mathcal{K}_{0}},d_{H})$ is a
compact metric space (for a proof we refer for example to \citet{Munkres}).
Let $\varepsilon$ be a positive real number, we consider the open ball of
radius $\varepsilon $ regarding to Hausdorff distance $\mathcal{B}%
_{\varepsilon }(\mathcal{A}_{0})$. The set $\boldsymbol{\mathcal{K}}_{1}=%
\boldsymbol{\mathcal{K}}_{0}\setminus \mathcal{B}_{\varepsilon }(\mathcal{A}%
_{0})$ is compact. Since $\varphi $ is continuous, it is well defined its
minimum over $\boldsymbol{\mathcal{K}}_{1}$, say $\nu _{1}$, namely, 
\begin{equation*}
\nu _{1}=\min_{\mathcal{A}\in \boldsymbol{\mathcal{K}}_{1}}\varphi (\mathcal{%
A}),
\end{equation*}%
provided $\nu _{0}$ is unique, must be $\nu _{1}>\nu _{0}$, hence we can
take $\eta =\nu _{1}-\nu _{0}>0$, the Lemma hypothesis ensures the existence
of $n_{0}$, for which, $\forall n\geq n_{0}$%
\begin{equation*}
\varphi (\mathcal{A}_{n})<(\nu _{1}-\nu _{0})+\nu _{0}=\nu _{1},
\end{equation*}%
then $\mathcal{A}_{n}$ can not belong to $\boldsymbol{\mathcal{K}}_{1}$,
since in that case it would be less than the minimum. Thus, $\mathcal{A}%
_{n}\in \mathcal{B}_{\varepsilon }(\mathcal{A}_{0})$
\end{proof}


\begin{lemma}
\label{unicaSolucionPob} Let $\mathcal{A}\subseteq \mathbb{R}^{p}$ be a set
of at most $K$ points. We consider $H:[0,+\infty )\rightarrow \mathbb{R}$ 
\begin{equation*}
H(t)=\left\{ 
\begin{array}{lrr}
1-\mathbb{P}(\mathbf{x}\in \mathcal{A}) & \mbox{si } & t=0 \\ 
&  &  \\ 
\mathbb{E}\left( \rho (\frac{d(\mathbf{x},\mathcal{A})}{t})\right) & %
\mbox{si } & t>0%
\end{array}%
\right.
\end{equation*}%
Then,

\begin{enumerate}
\item $H$ is continuous

\item $\lim_{t\rightarrow \infty }H(t)=0$

\item Equation $H(t)=1/2$ has unique solution.
\end{enumerate}
\end{lemma}

\begin{proof}
We will prove {\small 1.} First, we see the continuity for $t>0$, set $%
t_{0}>0$, and take $t_{j}\rightarrow t_{0}$. Let $(f_{j})_{j\in \mathbb{N}}$
be the functions sequence, $f_{j}(\mathbf{x}):=\rho (d(\mathbf{x},\mathcal{A}%
)/t_{j})$, $f_{j}$ are bounded by $1$, and converge pointwise to $\rho ({d(%
\mathbf{x},\mathcal{A})}/{t_{0}})$. Then, by using the Lebesgue dominated
convergence Theorem, it is possible to exchange the expectation and the
limit in the following equation, 
\begin{equation*}
\lim_{j\rightarrow \infty }H(t_{j})=\lim_{j\rightarrow \infty }\mathbb{E}%
\left( \rho \left( \frac{d(\mathbf{x},\mathcal{A})}{t_{j}}\right) \right) =%
\mathbb{E}\left( \lim_{j\rightarrow \infty }\rho \left( \frac{d(\mathbf{x},%
\mathcal{A})}{t_{j}}\right) \right) =H(t_{0}).
\end{equation*}%
In this way, $H$ is continuous at $t_{0}>0$. Now, it remains to see the continuity
at $t=0$, take $t_{j}$ converging decreasingly to $0$, 
\begin{equation*}
H(t_{j})=\mathbb{E}\left( \rho \left( \frac{d(\mathbf{x},\mathcal{A})}{t_{j}}%
\right) \right) =\mathbb{E}\left( \rho \left( \frac{d(\mathbf{x},\mathcal{A})%
}{t_{j}}\right) I_{\{d(\mathbf{x},\mathcal{A})>0\}}\right) .
\end{equation*}%
For $\rho$- function considered here $\rho (0)=0$, and also if $d(\mathbf{x},%
\mathcal{A})>0$, 
\begin{equation*}
\lim_{j\rightarrow \infty }\rho (\frac{d(\mathbf{x},\mathcal{A})}{t_{j}}%
)=\rho (\infty )=1.
\end{equation*}%
So, by applying dominated convergence Theorem, 
\begin{equation*}
\lim_{j\rightarrow \infty }H(t_{j})=\mathbb{E}\left( \lim_{j\rightarrow
\infty }\rho \left( \frac{d(\mathbf{x},\mathcal{A})}{t_{j}}\right) I_{\{d(%
\mathbf{x},\mathcal{A})>0\}}\right) =\mathbb{E}\left( 1\cdot I_{\{d(\mathbf{x%
},\mathcal{A})>0\}}\right) =1-\mathbb{P}(\mathbf{x}\in \mathcal{A}).
\end{equation*}%
Thus, $H(t)$ is continuous at $0$.


We will prove {\small 2}. Take $t_{j}\rightarrow \infty $, $\rho ({d(\mathbf{%
x},\mathcal{A})}/{t_{j}})\rightarrow 0$ pointwise, thus $H(t_{j})\rightarrow
0$ concluding the proof of {\small 2}.

Finally, we will see item 3 of the Lemma, in first place, {A.1)} Hypothesis
implies that $1-\mathbb{P}(\mathbf{x}\in \mathcal{A})>1/2$. We apply the
intermediate value Theorem to the function $H(t)$: $H(0)>\frac{1}{2}$, $%
H(\infty )=0$, then there exists $s$ such that $H(s)=1/2$, proving the existence,
to see the uniqueness, suppose that $s_{2}>s_{1}$ are two different
solutions, subtracting them, 
\begin{equation*}
0=H(s_{1})-H(s_{2})=\mathbb{E}\left( \rho \left( \frac{d(\mathbf{x},\mathcal{%
A})}{s_{1}}\right) -\rho \left( \frac{d(\mathbf{x},\mathcal{A})}{s_{2}}%
\right) \right) .
\end{equation*}%
Because $\rho $ is monotonous, the argument inside the expectation is
greater or equal than zero, then 
\begin{equation}
\mathbb{P}\left( \rho \left( \frac{d(\mathbf{x},\mathcal{A})}{s_{1}}\right)
-\rho \left( \frac{d(\mathbf{x},\mathcal{A})}{s_{2}}\right) =0\right) =1,
\label{proba1}
\end{equation}%
Thus, 
\begin{equation*}
\rho \left( \frac{d(\mathbf{x},\mathcal{A})}{s_{1}}\right) =\rho \left( 
\frac{d(\mathbf{x},\mathcal{A})}{s_{2}}\right)
\end{equation*}%
all most everywhere. That may happen in two ways:

\begin{itemize}
\item ${d(\mathbf{x},\mathcal{A})}/{s_{1}}={d(\mathbf{x},\mathcal{A})}/{s_{2}%
}$, but given that $s_{1}\neq s_{2}$, the previous equation only happen if $%
d(\mathbf{x},\mathcal{A})=0$, that means that $\mathbf{x}\in \mathcal{A}$.

\item ${d(\mathbf{x},\mathcal{A})}/{s_i} \in [m,+\infty)$, being $m$ the
value from which $\rho(x)=1 \forall x\geq m$.
\end{itemize}

Then, we write equation (\ref{proba1}) according to the previous two events
described, and get 
\begin{equation*}
\mathbb{P}\left( \mathbf{x}\in \mathcal{A}\right) +\mathbb{P}\left( \frac{d(%
\mathbf{x},\mathcal{A})}{s_{1}}\in \lbrack m,+\infty )\right) =1.
\end{equation*}%
from the equation above, and using hypothesis {A.1)} it turns out $\mathbb{P}%
\left( {d(\mathbf{x},\mathcal{A})}/{s_{1}}\in \lbrack m,+\infty )\right)
>1/2.$ Finally, we take the $H$ definition and come to a contradiction 
\begin{equation*}
\frac{1}{2}=H(s_{1})=\mathbb{E}\left( \rho \left( \frac{d(\mathbf{x},%
\mathcal{A})}{s_{1}}\right) \right) \geq \mathbb{E}\left( \rho \left( \frac{%
d(\mathbf{x},\mathcal{A})}{s_{1}}\right) I_{\{\frac{d(\mathbf{x},\mathcal{A})%
}{s_{1}}\in \lbrack m,+\infty )\}}\right)
\end{equation*}%
\begin{equation*}
=\mathbb{P}\left( \frac{d(\mathbf{x},\mathcal{A})}{s_{1}}\in \lbrack
m,+\infty )\right) >\frac{1}{2}.
\end{equation*}%
The absurdity was caused by supposing that there were two different
solutions $s_{1}$ and $s_{2}$.
\end{proof}



\begin{lemma}
\label{lemaContinuidad} 

Functions $\varphi _{\tau }(\mathcal{A})=\tau (\mathcal{A},F)$ and $\varphi
_{M}(\mathcal{A})=M(\mathcal{A},F)$, are continuous regarding to Hausdorff
distance. 
\end{lemma}

\begin{proof}
Let $\mathcal{A}_{n}$ be a sequence converging to $\mathcal{A}$ in the
Hausdorff sense, we will see that $\lim_{n\rightarrow \infty }\varphi _{M}(%
\mathcal{A}_{n})=\varphi _{M}(\mathcal{A})$. Take $\varepsilon >0$, it is
easy to see that for each $\mathbf{x}$, $d(\mathbf{x},\mathcal{A}%
_{n})\rightarrow d(\mathbf{x},\mathcal{A})$, then, by using dominated
convergence Theorem, it turns out that 
\begin{equation*}
\mathbb{E}\left( \rho _{1}\left( \frac{d(\mathbf{x},\mathcal{A}_{n})}{M(%
\mathcal{A},F)+\varepsilon }\right) \right) \rightarrow \mathbb{E}\left(
\rho _{1}\left( \frac{d(\mathbf{x},\mathcal{A})}{M(\mathcal{A}%
,F)+\varepsilon }\right) \right) .
\end{equation*}%
On the other hand, 
\begin{equation*}
\mathbb{E}\left( \rho _{1} \left(\frac{d(\mathbf{x},\mathcal{A})}{M(\mathcal{A}%
,F)+\varepsilon }\right)\right) <\frac{1}{2},
\end{equation*}%
then, there exists $n_{0}$ such that $\forall n\geq n_{0}$

\begin{equation}
\mathbb{E}\left( \rho _{1}\left( \frac{d(\mathbf{x},\mathcal{A}_{n})}{M(%
\mathcal{A},F)+\varepsilon }\right) \right) <\frac{1}{2}.
\label{ecu1Intermedio}
\end{equation}
Analogously it can be shown that exists $n_{1}$, such that $\forall n\geq
n_{1}$ 
\begin{equation}
\mathbb{E}\left( \rho _{1}\left( \frac{d(\mathbf{x},\mathcal{A}_{n})}{M(%
\mathcal{A},F)-\varepsilon }\right) )\right) >\frac{1}{2},
\label{ecu2Intermedio}
\end{equation}%
For $n\geq \max \{n_{1},n_{2}\}$, consider the function $H_{n}(t)=\mathbb{E}%
\left( \rho _{1}(\frac{d(\mathbf{x},\mathcal{A}_{n})}{t})\right) $, we apply
Lemma \ref{unicaSolucionPob}, obtaining continuity of $H_{n}(t)$ and
uniqueness for the problem 
\begin{equation*}
\mathbb{E}\left( \rho _{1}\left( \frac{d(\mathbf{x},\mathcal{A}_{n})}{%
t^{\ast }}\right) \right) =\frac{1}{2}
\end{equation*}%
Let $t^{\ast }=M(\mathcal{A}_{n},F)$, through intermediate value theorem for 
$H_{n}(t)$, inequalities (\ref{ecu1Intermedio}) and (\ref{ecu2Intermedio}),
together with uniqueness shown in \ref{unicaSolucionPob}, it is easy to see
that $t^{\ast }\in (M(\mathcal{A},F)-\varepsilon ,M(\mathcal{A}%
,F)+\varepsilon )$. Then, 
\begin{equation*}
|M(\mathcal{A}_{n},F)-M(\mathcal{A},F)|<\varepsilon ,
\end{equation*}%
thus, $\varphi _{M}(\mathcal{A})=M(\mathcal{A},F)$ is continuous as function
of $\mathcal{A}$. 

Now, we need to prove the continuity of $\varphi_{\tau }(\mathcal{A})$, 
\begin{equation*}
\varphi_{\tau }(\mathcal{A})^{2}=(M(A,F))^{2}\mathbb{E}\left( \rho
_{2}\left( \frac{d(\mathbf{x},\mathcal{A})}{M(\mathcal{\ A},F)}\right)
\right)
\end{equation*}%
Take $\mathcal{A}_{n}$ converging to $\mathcal{A}$ in Hausdorff distance,
given that $M(\mathcal{A},F)>0$, and $M(\mathcal{A},F)$ is continuous in $%
\mathcal{A}$, the integrand is continuous and bounded. By using the
dominated convergence Theorem we get that $\tau (\mathcal{A}_{n},F)^{2}$
converges to $\tau ^{2}(\mathcal{A},F)$. Thus, $\varphi _{\tau }(\mathcal{A}%
) $ is continuous as function of $\mathcal{A}$.
\end{proof}


\begin{lemma}
{(Uncoupled Uniform Convergence over Compact Sets)} \label{ConvUnifDesacop}
Consider the set $\boldsymbol{\mathcal{K}}_{0}=\{\mathcal{A}\subseteq 
\overline{B},\#\mathcal{A}\leq K\}$, where $\overline{B}$ is a closed ball
in $\mathbb{R}^{p}$, let $\rho $ be a $\rho $- function and $l_{2}>l_{1}>0$.
Let $\mathbf{x}_{1},\mathbf{x}_{2},...,\mathbf{x}_{n}$ independent
observations from a random sample of size $n$. Then 
\begin{equation*}
\lim_{n\rightarrow \infty }\sup_{\mathcal{A}\in \boldsymbol{\mathcal{K}}%
_{0},s\in \lbrack l_{1},l_{2}]}\left\vert \frac{1}{n}\sum_{i=1}^{n}\rho
\left( \frac{d(\mathbf{x}_{i},\mathcal{A})}{s}\right) -E_{F}\left( \rho
\left( \frac{d(\mathbf{x},\mathcal{A})}{s}\right) \right) \right\vert =0\ \
\ \text{a.s.}
\end{equation*}%
%
%
%
%
%
\end{lemma}

\begin{proof}
We consider the Family 
\begin{equation*}
\mathcal{G}=\left\{ g_{\mathcal{A},s}(\mathbf{x})=\rho \left( \frac{d(%
\mathbf{x},\mathcal{A})}{s}\right) ,\mathcal{A}\in \boldsymbol{\mathcal{K}%
_{0}},s\in \lbrack l_{1},l_{2}]\right\} .
\end{equation*}%
We want to prove an uniform strong law of large numbers (USLLN) for $%
\mathcal{G}$. Sufficient conditions for the theorem to hold are given in 
\citet{PollardLibro}. In particular, it is established that if for each $%
\varepsilon >0$, there exists a finite family $\mathcal{F}_{\varepsilon }$
satisfying 
\begin{equation*}
\forall g\in \mathcal{G}\ \ \exists \ \ f_{1},f_{2}\in \mathcal{F}%
_{\varepsilon },\mbox{ such that }f_{1}\leq g\leq f_{2}\mbox{ and also }%
\mathbb{E}(f_{2}-f_{1})\leq \varepsilon ,
\end{equation*}%
then, family $\mathcal{G}$ has a USLLN.

Given $\varepsilon >0$, we show a $\mathcal{F}_\varepsilon$ that satisfies
this property.

First, because uniform continuity of $\rho $ function, we can choose $\delta
>0$ such that

\begin{equation}
|\rho \left( y+\delta \right) -\rho \left( y\right) |<\frac{1}{3}\varepsilon
\quad \quad \forall y\in \mathbb{R}  \label{detallando1}
\end{equation}%
Let $\delta _{1}$ be such that $0<\delta _{1}\leq {l_{1}}\delta $, as $%
\overline{B}\subseteq \mathbb{R}^{p}$ is a closed ball, it is possible to
take a finite subset $\mathcal{J}$ whose elements are $\mathbf{a}_{j}\in 
\overline{B}$, con $1\leq j\leq N$, $\mathcal{J}=\{\mathbf{a}_{1},\dots ,%
\mathbf{a}_{N}\}$, such that $\forall \mathbf{x}\in \overline{B},\exists
j:\Vert \mathbf{x}-\mathbf{a}_{j}\Vert <\delta _{1}$. Let $0<\delta
_{2}<\varepsilon {l_{1}}/(3C_{\psi })$ be a real number, where $C_{\psi
}=\sup_{u\in \mathbb{R}}\psi (u)u$ is a positive number for the $\rho $
function considered here. Consider the partition $[l_{1},l_{2}]=\cup
_{i=1}^{M}[s_{i-1},s_{i}]$ with interval length $[s_{i-1},s_{i}]$ less than $%
\delta _{2}$. Family $\mathcal{F}_{\varepsilon }$ will be 
\begin{equation}
\mathcal{F}_{\varepsilon }=\left\{ \rho \left( \frac{d(\mathbf{x},\mathcal{A}%
^{\prime })}{s_{i-1}}\pm \delta \right) :\mathcal{A}^{\prime }\subseteq 
\mathcal{J},\#(A^{\prime })\leq k,1\leq i\leq M\right\} .
\end{equation}

Now set an element of $\mathcal{G}$ indexed by $\mathcal{A}$ and $s$, then $%
s\in \lbrack s_{i-1},s_{i}]$, for some $i$, and also there exists a $%
\mathcal{A}^{\prime }\subseteq \mathcal{J}$, such that $d_{H}(\mathcal{A},%
\mathcal{A}^{\prime })\leq \delta _{1}.$ Then

\begin{equation}
\frac{d(\mathbf{x},\mathcal{A})}{s}\leq \frac{d(\mathbf{x},\mathcal{A}%
^{\prime })+\delta _{1}}{s_{i-1}}\leq \frac{d(\mathbf{x},\mathcal{A}^{\prime
})}{s_{i-1}}+\frac{\delta _{1}}{l_{1}},  \label{desAAprime1}
\end{equation}%
as consequence,

\begin{equation}
\frac{d(\mathbf{x},\mathcal{A})}{s}\geq \frac{d(\mathbf{x},\mathcal{A}%
^{\prime })-\delta _{1}}{s_{i}}\geq \frac{d(\mathbf{x},\mathcal{A}^{\prime })%
}{s_{i}}-\frac{\delta _{1}}{l_{2}}\geq \frac{d(\mathbf{x},\mathcal{A}%
^{\prime })}{s_{i}}-\frac{\delta _{1}}{l_{1}},  \label{desAAprime2}
\end{equation}%
Provided $\delta _{1}/l_{1}\leq \delta $, and using the $\rho $ monotonicity
at (\ref{desAAprime1}) and (\ref{desAAprime2}), we get

\begin{equation*}
f_{1}(\mathbf{x})\leq \rho \left( \frac{d(\mathbf{x},\mathcal{A})}{s}\right)
\leq f_{2}(\mathbf{x}),
\end{equation*}%
where 
\begin{equation*}
f_{1}(\mathbf{x})=\rho \left( \frac{d(\mathbf{x},\mathcal{A}^{\prime })}{%
s_{i}}-\delta \right) ,\quad \text{and}\quad f_{2}(\mathbf{x})=\rho \left( 
\frac{d(\mathbf{x},\mathcal{A}^{\prime })}{s_{i-1}}+\delta \right) .
\end{equation*}

It remain to see that choices made on $\delta $, $\delta _{1}$ and $\delta
_{2}$ , imply $\mathbb{E}\left( f_{2}-f_{1}\right) \leq \varepsilon $.
Indeed, 
\begin{equation*}
\begin{array}{rcl}
\mathbb{E}\left( f_{2}-f_{1}\right) & = & \mathbb{E}\left( \rho \left( \frac{%
d(\mathbf{x},\mathcal{A}^{\prime })}{s_{i}}+\delta \right) -\rho \left( 
\frac{d(\mathbf{x},\mathcal{A}^{\prime })}{s_{i-1}}-\delta \right) \right)
\\ 
&  &  \\ 
& = & \mathbb{E}\left( \rho \left( \frac{d(\mathbf{x},\mathcal{A}^{\prime })%
}{s_{i}}+\delta \right) -\rho \left( \frac{d(\mathbf{x},\mathcal{A}^{\prime
})}{s_{i}}\right) \right) +\mathbb{E}\left( \rho \left( \frac{d(\mathbf{x},%
\mathcal{A}^{\prime })}{s_{i-1}}\right) -\rho \left( \frac{d(\mathbf{x},%
\mathcal{A}^{\prime })}{s_{i-1}}-\delta \right) \right) + \\ 
&  &  \\ 
& + & \mathbb{E}\left( \rho \left( \frac{d(\mathbf{x},\mathcal{A}^{\prime })%
}{s_{i}}\right) -\rho \left( \frac{d(\mathbf{x},\mathcal{A}^{\prime })}{%
s_{i-1}}\right) \right) .%
\end{array}%
\end{equation*}%
Applying inequality (\ref{detallando1}) with $y={d(\mathbf{x},\mathcal{A}%
^{\prime })}/{s_{i}}$ in their first term, and with $y={d(\mathbf{x},%
\mathcal{A}^{\prime })}/{s_{i-1}}$ in the second, the absolute value of the
first two terms are less than ${2\varepsilon }/{3}$. To bound the last one,
by middle value theorem, there exists $\xi _{i}^{\ast }\in \lbrack
s_{i-1},s_{i}]$ satisfying 
\begin{equation*}
\rho \left( \frac{d(\mathbf{x},\mathcal{A}^{\prime })}{s_{i}}\right) -\rho
\left( \frac{d(\mathbf{x},\mathcal{A}^{\prime })}{s_{i-1}}\right) =-\psi
\left( \frac{d(\mathbf{x},\mathcal{A}^{\prime })}{\xi ^{\ast }}\right) \frac{%
d(\mathbf{x},\mathcal{A}^{\prime })}{\xi _{i}^{\ast }}\frac{1}{\xi
_{i}^{\ast }}(s_{i}-s_{i-1}).
\end{equation*}%
As $\psi (u)u$ is bounded by $C_{\psi }$ it turns out 
\begin{equation*}
\mathbb{E}\left( \left\vert \rho \left( \frac{d(\mathbf{x},\mathcal{A}%
^{\prime })}{s_{i}}\right) -\rho \left( \frac{d(\mathbf{x},\mathcal{A}%
^{\prime })}{s_{i-1}}\right) \right\vert \right) \leq \frac{C_{\psi }}{l_{1}}%
\delta _{2}.
\end{equation*}%
Then, as $\delta _{2}$ was chosen in such way that $\delta _{2}\leq {%
l_{1}\varepsilon }/({3C_{\psi }})$, ${\varepsilon }/{3}$ is a properly bound
for the previous term. Thus, we get 
\begin{equation*}
\mathbb{E}\left( f_{2}-f_{1}\right) \leq \frac{2}{3}\varepsilon +\frac{1}{3}%
\varepsilon =\varepsilon
\end{equation*}
\end{proof}
\begin{lemma}
{\textit{(Uniform Strong Law of Large Numbers)}} \label{ConvUnifacop}
Consider the set $\boldsymbol{\mathcal{K}}_0=\{ \mathcal{A}\subseteq 
\overline{B}, \#\mathcal{A} \leq K\}$, where $\overline{B}$ is a closed ball
in $\mathbb{R}^p$. Suppose A.1) and $(\mathbf{x}_{n})_{n\in\mathbb{N}}$ are $%
i.i.d^{\prime}s$, then: 
\begin{equation}
\lim_{n\rightarrow\infty}\sup_{\mathcal{A}\in\boldsymbol{\mathcal{K}}%
_0}\left \vert \tau(\mathcal{A},F)-\tau(\mathcal{A},F_{n})\right\vert =0\ \
\ a.s.  \label{tauunif}
\end{equation}
\end{lemma}

\begin{proof}
Consider $h_{1}=\inf_{\mathcal{A}\in \boldsymbol{\mathcal{K}}_{0}}M(\mathcal{%
A},F)$ and $h_{2}=\sup_{\mathcal{A}\in \boldsymbol{\mathcal{K}}_{0}}M(%
\mathcal{A},F)$. Take $\varepsilon <{h_{1}}/{2}$. Define $g_{1}$ and $g_{2}$
as follows

\begin{equation}
g_{1}(\mathcal{A})=\mathbb{E}\left( \rho _{1}\left( \frac{d(\mathbf{x},%
\mathcal{A})}{M(\mathcal{A},F)-\varepsilon }\right) \right)
\end{equation}%
and%
\begin{equation}
g_{2}(\mathcal{A})=\mathbb{E}\left( \rho _{1}\left( \frac{d(\mathbf{x},%
\mathcal{A})}{M(\mathcal{A},F)+\varepsilon }\right) \right) .
\end{equation}%
As $g_{1}$ and $g_{2}$ are continuous functions regarding to $d_{H}$, and $%
\boldsymbol{\mathcal{K}}_{0}$ is a compact set under $d_{H}$, minimum and
maximum are achieved at $\mathcal{A}_{1}$ and $\mathcal{A}_{2}$
respectively, whose values are $g_{\ell }(\mathcal{A}_{\ell })=d_{\ell
},\ell =1,2$. Define $d_{1}$ and $d_{2}$ as $\inf_{\mathcal{A}\in 
\boldsymbol{\mathcal{K}}_{0}}g_{1}(\mathcal{A})=d_{1}>{1}/{2}$, and $\sup_{%
\mathcal{A}\in \boldsymbol{\mathcal{K}}_{0}}g_{2}(\mathcal{A})=d_{2}<{1}/{2}$%
.

Taking a real number $\delta $, $0<\delta \leq (d_{1}-{1}/{2})$ and $%
0<\delta \leq ({1}/{2}-d_{2})$, and through Lemma \ref{ConvUnifDesacop}, $%
\exists \Omega ^{\prime },\mathbb{P}(\Omega ^{\prime })=1$, satisfying $%
\forall \omega \in \Omega ^{\prime }\exists n_{0}=n_{0}(\omega )$ such that
for all $n>n_{0}$

\begin{equation*}
\sup_{\mathcal{A} \in \boldsymbol{\mathcal{K}}_0, s \in [\frac{h_1}{2},
h_2]} \left| \frac{1}{n}\sum_{i=1}^{n}\rho_1\left(\frac{d(\mathbf{x}_{i},%
\mathcal{A})}{s}\right)- \mathbb{E}\left(\rho_1\left(\frac{d(\mathbf{x},%
\mathcal{A})}{s}\right)\right) \right| < \frac {\delta}{2},
\end{equation*}
where $(\mathbf{x}_i^{(\omega)})_{i\in \mathbb{N}}$ depends on $\omega$, but
we write $(\mathbf{x}_i)_{i\in \mathbb{N}}$ for short.

Suppose that sequence $\mathcal{A}\in \boldsymbol{\mathcal{K}}_{0}$, then $M(%
\mathcal{A},F)-\varepsilon \in (h_{1}-\varepsilon ,h_{2}-\varepsilon
)\subseteq \lbrack {h_{1}}/{2},h_{2}]$, thus

\begin{equation*}
\sup_{\mathcal{A}\in \boldsymbol{\mathcal{K}}_{0}}\left\vert \frac{1}{n}%
\sum_{i=1}^{n}\rho _{1}\left( \frac{d(\mathbf{x}_{i},\mathcal{A})}{M(%
\mathcal{A},F)-\varepsilon }\right) -\mathbb{E}\left( \rho _{1}\left( \frac{%
d(\mathbf{x},\mathcal{A})}{M(\mathcal{A},F)-\varepsilon }\right) \right)
\right\vert <\frac{\delta }{2}.
\end{equation*}%
Then, the follow inequality is valid for all $\mathcal{A}\in \boldsymbol{%
\mathcal{K}}_{0}$, 
\begin{equation*}
\frac{1}{n}\sum_{i=1}^{n}\rho _{1}\left( \frac{d(\mathbf{x}_{i},\mathcal{A})%
}{M(\mathcal{A},F)-\varepsilon }\right) >\mathbb{E}\left( \rho _{1}\left( 
\frac{d(\mathbf{x},\mathcal{A})}{M(\mathcal{A},F)-\varepsilon }\right)
\right) -\frac{\delta }{2}.
\end{equation*}%
So, by taking infimum at the right hand side 
\begin{equation*}
\frac{1}{n}\sum_{i=1}^{n}\rho _{1}\left( \frac{d(\mathbf{x}_{i},\mathcal{A})%
}{M(\mathcal{A},F)-\varepsilon }\right) \geq \inf_{\mathcal{A}\in 
\boldsymbol{\mathcal{K}}_{0}}\mathbb{E}\left( \rho _{1}\left( \frac{d(%
\mathbf{x},\mathcal{A})}{M(\mathcal{A},F)-\varepsilon }\right) \right) -%
\frac{\delta }{2},
\end{equation*}%
therefore 
\begin{equation*}
\ \inf_{\mathcal{A}\in \boldsymbol{\mathcal{K}}_{0}}\frac{1}{n}%
\sum_{i=1}^{n}\rho _{1}\left( \frac{d(\mathbf{x}_{i},\mathcal{A})}{M(%
\mathcal{A},F)-\varepsilon }\right) \geq \inf_{\mathcal{A}\in \boldsymbol{%
\mathcal{K}}_{0}}g_{1}(\mathcal{A})-\frac{\delta }{2}=d_{1}-\frac{\delta }{2}%
\geq \frac{1}{2}+\frac{\delta }{2}.
\end{equation*}%
Then we obtain 
\begin{equation}
\inf_{\mathcal{A}\in \boldsymbol{\mathcal{K}}_{0}}\frac{1}{n}%
\sum_{i=1}^{n}\rho _{1}\left( \frac{d(\mathbf{x}_{i},\mathcal{A})}{M(%
\mathcal{A},F)-\varepsilon }\right) \geq \frac{1}{2}+\frac{\delta }{2}.
\label{inf1}
\end{equation}%
Analogously 
\begin{equation}
\sup_{\mathcal{A}\in \boldsymbol{\mathcal{K}}_{0}}\frac{1}{n}%
\sum_{i=1}^{n}\rho _{1}\left( \frac{d(\mathbf{x}_{i},\mathcal{A})}{M(%
\mathcal{A},F)+\varepsilon }\right) \leq \frac{1}{2}-\frac{\delta }{2}
\label{sup1}
\end{equation}%
Now consider the function at the $t$ variable, 
\begin{equation*}
H_{\mathcal{A}}^{n}(t)=\frac{1}{n}\sum_{i=1}^{n}\rho _{1}\left( \frac{d(%
\mathbf{x}_{i},\mathcal{A})}{t}\right) ,
\end{equation*}%
By using (\ref{inf1}) and (\ref{sup1}), by intermediate values theorem, there exists $t^{\ast }\in I$, such that $H_{\mathcal{A}}^{n}(t^{\ast })=%
\frac{1}{2}$, where $I=(M(\mathcal{A},F)-\varepsilon ,M(\mathcal{A}%
,F)+\varepsilon )$, thus, by the uniqueness of $M$ scale, $t^{\ast }=M(%
\mathcal{A},F_{n})$, and we get that 
\begin{equation*}
|M(\mathcal{A},F_{n})-M(\mathcal{A},F)|<\varepsilon .
\end{equation*}%
As $\varepsilon $ is an upper bound, non dependent of $\mathcal{A}$, it is
possible to take supremum over $\mathcal{A}\in \boldsymbol{\mathcal{K}}_{0}$%
, an we obtain 
\begin{equation*}
\sup_{\mathcal{A}\in \boldsymbol{\mathcal{K}}_{0}}|M(\mathcal{A},F_{n})-M(%
\mathcal{A},F)|<\varepsilon
\end{equation*}
\end{proof}

The following lemma says that if $u_{1}<u_{2}$ are two positive numbers
separated from each other, the only way for $\rho (u_{2})-\rho(u_{1})$ to be
arbitrary small, is that $\rho (u_{1})$ to be arbitrary close to one. 

\begin{lemma}
\label{propiedadderho} Let $\rho $ be a $\rho$-function strictly increasing
in the interval $[0,q)$ and $\rho (u)=1$ for $u\geq q$. Let $u\geq \alpha >0$
and $\Delta \geq t>0$ be real positive numbers, then $\forall \kappa \in
(0,1)$ $\exists \gamma =\gamma (\kappa ,\alpha ,t)$ such that 
\begin{equation*}
\rho (u+\Delta )-\rho (u)\leq \gamma \Rightarrow \rho (u)>1-\kappa .
\end{equation*}
\end{lemma}

\begin{proof}
Let $\kappa \in (0,1)$ consider 
\begin{equation*}
\ell =\inf \{\rho (u+\Delta )-\rho (u):\ \ \alpha \leq u\leq \rho
^{-1}(1-\kappa ),\ \ \Delta \geq t\}.
\end{equation*}%
If $\alpha >\rho ^{-1}(1-\kappa )$ we have $\rho (u)>1-\kappa $ regardless
the choice of $\gamma $, so there is nothing to prove. We can suppose $\alpha
\leq \rho ^{-1}(1-\kappa )$, the subset who is being taken infimum is a non
empty set, and lower bounded by $0$. Assume that $\ell >0$, take $\gamma
=\ell /2$, we will see that the Lemma holds, Let $u\geq \alpha >0$ and $%
\Delta \geq t>0$ two numbers such that $\rho _{1}(u+\Delta )-\rho
_{1}(u)\leq \gamma =\ell /2$, then if $u\leq \rho ^{-1}(1-\kappa )$, there
would be an element belonging to the set considered but also lower than the
infimum, that is a contradiction, therefore $u>\rho ^{-1}(1-\kappa )$, so $%
\rho (u)>1-\kappa $.

To finish the proof, we shall see that $\ell >0$, if $\ell =0$ take $\rho
(u_{n}+\Delta _{n})-\rho (u_{n})\rightarrow \ell =0$. First of all suppose
that $(\Delta _{n})_{n\in \mathbb{N}}$ it is not bounded, then we can find
subsequences such that $\Delta _{n_{j}}\rightarrow \infty $ and $%
u_{n_{j}}\rightarrow u^{\ast }\in \lbrack \alpha ,\rho ^{-1}(1-\kappa )]$.
So, 
\begin{equation*}
0=\ell =\lim_{j\rightarrow \infty }\rho (u_{n_{j}}+\Delta _{n_{j}})-\rho
(u_{n_{j}})=1-\rho (u^{\ast })\geq 1-(1-\kappa )=\kappa >0,
\end{equation*}%
then a contradiction is caused because $\ell =0$ and $(\Delta _{n})_{n\in 
\mathbb{N}}$ is non bounded. Now, we will see that neither occurs $(\Delta
_{n})_{n\in \mathbb{N}}$ is bounded and $\ell =0$. If so, choose
subsequences $\Delta _{n_{j}}\rightarrow \Delta ^{\ast }\geq t$ and $%
u_{n_{j}}\rightarrow u^{\ast }<\rho ^{-1}(1-\kappa )$, then 
\begin{equation*}
0=\ell =\lim_{j\rightarrow \infty }\rho (u_{n_{j}}+\Delta _{n_{j}})-\rho
(u_{n_{j}})=\rho (u^{\ast }+\Delta ^{\ast })-\rho (u^{\ast }).
\end{equation*}%
Thus, $\rho (u^{\ast }+\Delta ^{\ast })=\rho (u^{\ast })$, this can happen
in two ways

\begin{enumerate}
\item Both $u^{\ast }$ and $u^{\ast }+\Delta ^{\ast }$ are in an interval
where $\rho $ is constant, that is impossible because $u^{\ast }\leq \rho
^{-1}(1-\kappa ).$

\item $u^{\ast }=u^{\ast }+\Delta ^{\ast }$, then $\Delta ^{\ast }=0$ which
is absurd since $\Delta ^{\ast }\geq t>0$.
\end{enumerate}

Finally, the case $\ell=0$ have been discarded, that it was what we wanted
to prove
\end{proof}


\begin{lemma}
\label{LemaPollardhcero} Define 
\begin{equation*}
\mathcal{B}_{\mathcal{A}}(\delta )=\left\{ \mathbf{x}:d(\mathbf{x},\mathcal{A%
})\leq \delta \right\} ,
\end{equation*}%
where $\mathbf{x}\in \mathbb{R}^{p}$, $\mathcal{A}\subseteq \mathbb{R}^{p}$
is a set of at most $K$ points, and $d(\mathbf{x},\mathcal{A})=\min_{%
\boldsymbol{\mu }\in \mathcal{A}}\Vert \mathbf{x}-\boldsymbol{\mu }\Vert $,
then, (i)%
\begin{equation*}
\lim_{n\rightarrow \infty }\sup_{\mathcal{A},\#\mathcal{A}\leq K}\left\vert 
\mathbb{P}_{n}(\mathcal{B}_{\mathcal{A}}(\delta ))-\mathbb{P}(\mathcal{B}_{%
\mathcal{A}}(\delta ))\right\vert =0\text{ \ a.s.}
\end{equation*}%
and (ii)%
\begin{equation*}
\lim_{n\rightarrow \infty }\inf_{\mathcal{A},\#\mathcal{A}\leq K}\mathbb{P}%
_{n}(\mathcal{B}_{\mathcal{A}}(\delta )^{c})=\inf_{\mathcal{A},\#\mathcal{A}%
\leq K}\mathbb{P}(\mathcal{B}_{\mathcal{A}}(\delta )^{c})\text{ \ a.s.}
\end{equation*}%
where, for a finite sample of size $n$, $\mathbb{P}_{n}(\mathcal{C})=\#(\{%
\mathbf{x}_{i}\in \mathcal{C}\})/n$.

\end{lemma}

\begin{proof}

First, we note that $\mathcal{B}_{\mathcal{A}}(\delta )$ can be written as

\begin{equation*}
\mathcal{B}_{\mathcal{A}}(\delta )=\cup _{\boldsymbol{\mu }\in \mathcal{A}}%
\overline{B}(\boldsymbol{\mu },\delta ),
\end{equation*}%
Therefore, defining the collection of sets $\mathcal{G}$ formed by sets of
at most $K$ closed balls of radius $\delta $ 
\begin{equation*}
\mathcal{G}^{K}=\left\{ \cup _{\boldsymbol{\mu }\in \mathcal{A}}\overline{B}(%
\boldsymbol{\mu },\delta ):\mathcal{A}\subseteq \#\mathcal{A}\leq K\right\} ,
\end{equation*}%
\'{\i}tem (i) of Lemma is expressed as Strong Law of Large Numbers in the
following way: 
\begin{equation*}
\lim_{n\rightarrow \infty }\sup_{g\in \mathcal{G}^{K}}\left\vert \mathbb{P}%
_{n}(g)-\mathbb{P}(g)\right\vert =0\text{ a.s.}
\end{equation*}

We rewrite the class as $\mathcal{G}^{K}=\{\cup _{j=1}^{k_{1}}g_{j}:g_{j}\in 
\mathcal{H},0\leq k_{1}\leq K\}$, where $\mathcal{H}$ is the class of closed
balls centered at $\boldsymbol{\mu }$ with fixed radius $\delta $. From
Theorem (14), and Lemmas (15) and (18) de \citet{PollardLibro}, it is
possible to see that if for each element $g_{i}$ of class $\mathcal{H}$
there exists a function $h\in \mathbb{V}$, such that $g_{j}=\{\mathbf{x}:h(%
\mathbf{x})\leq 0\}$, with $\mathbb{V}$ a function vector space with finite
dimension. Then it is valid an Uniform Strong Law of Large Numbers for $%
\mathcal{G}^{K}$.

Let $g_{j}\in \mathcal{H}$ be an element from $\mathcal{H}$, we will see
that there exists a function $h$ like was described previously. Let $g_{j}=\{%
\mathbf{x}\in B(\boldsymbol{\mu _{j}},\delta )\}$, then, choosing $h(\mathbf{%
x})=\Vert \mathbf{x}-\boldsymbol{\mu }_{j}\Vert ^{2}-\delta ^{2}$, all $%
g_{j}\in \mathcal{H}$ is $g_{j}=\{\mathbf{x}:h(\mathbf{x})\leq 0\}$ with $h$
a multivariate polynomial of degree at most $2$, given that polynomials are
a vector space of finite dimension, it is proven (i). Part (ii) of Lemma is
derived directly from (i).
\end{proof}


\begin{lemma}

\label{escalaAcotada} Let $\mathcal{A}_{n}$ a sequence of sets of $K$
centers, such that there exists $m^{\ast }>0$ for that the set 
\begin{equation*}
\Omega _{0}=\{\omega \in \Omega :\lim \sup M(\mathcal{A}_{n},F_{n}^{(\omega
)})<m^{\ast }\}
\end{equation*}%
has $\mathbb{P}(\Omega _{0})=1$. Then there exists $R_{1}$ and $\Omega
^{\prime }\subseteq \Omega _{0}$ of probability one satisfying 
\begin{equation}
\forall \omega \in \Omega ^{\prime }\text{ }\exists n_{0}(\omega ):si\ n\geq
n_{0}\quad \Rightarrow \mathcal{A}_{n}\cap B(R_{1})\neq \emptyset .
\label{desired}
\end{equation}

\end{lemma}

\begin{proof}
Take a real number $R^{\prime }$ satisfying simultaneously (a) $R^{\prime
}\geq \rho _{1}^{-1}(3/4)m^{\ast }$, and (b) $\mathbb{P}(B(R^{\prime }))>2/3$%
. Define $\Omega _{1}$ the set of probability $1$, where $\mathbb{P}%
_{n}(B(R^{\prime }))\rightarrow \mathbb{P}(B(R^{\prime }))$. Let $%
R_{1}=2R^{\prime }$, we will show that thesis lemma occurs in the set $%
\Omega _{0}\cap \Omega _{1}$. We define the set $\Omega ^{\ast }$ by denying
what we want to prove 
\begin{equation*}
\Omega ^{\ast }=\{\omega \in \Omega _{0}:\forall n\quad \exists n_{0}(\omega
)\geq n\quad \mbox{satisfying}\mathcal{A}_{n_{0}}\cap B(R_{1})=\emptyset \},
\end{equation*}%
\begin{equation*}
(\Omega _{0}\cap \Omega _{1})=(\Omega _{0}\cap \Omega _{1}\cap \Omega ^{\ast
})\cup (\Omega _{0}\cap \Omega _{1}\cap {\Omega ^{\ast }}^{c}).
\end{equation*}%
We will see that $(\Omega _{0}\cap \Omega _{1}\cap \Omega ^{\ast
})=\emptyset $, therefore defining $\Omega ^{\prime }=(\Omega _{0}\cap
\Omega _{1}\cap {\Omega ^{\ast }}^{c})$, we have that $\mathbb{P}(\Omega
^{\prime })=1$, $\Omega ^{\prime }$ is the set that satisfies the desired
property (\ref{desired}).

Proof of $(\Omega _{0}\cap \Omega _{1}\cap \Omega ^{\ast })=\emptyset $:
Suppose that $\omega \in (\Omega _{0}\cap \Omega _{1}\cap \Omega ^{\ast })$,
then given that $\omega \in \Omega ^{\ast }$, there exists a subsequence $%
\mathcal{A}_{n_{j}}$, such that $\mathcal{A}_{n_{j}}\cap B(R_{1})=\emptyset $%
. In particular, if $\boldsymbol{\mu }\in \mathcal{A}_{n_{j}}$, then $\Vert 
\boldsymbol{\mu }\Vert \geq R_{1}=2R^{\prime }$. Let $\mathbf{x}_{i}\in
B(R^{\prime })$ be a fixed point, then distance $d(\mathbf{x}_{i},\mathcal{A}%
_{n_{j}})$ it is achieved for some $\boldsymbol{\mu }_{n_{j}}^{i}$.
Therefore, 
\begin{equation*}
d(\mathbf{x}_{i},\mathcal{A}_{n_{j}})=\Vert \mathbf{x}_{i}-\boldsymbol{\mu }%
_{n_{j}}^{i}\Vert \geq \Vert \boldsymbol{\mu }_{n_{j}}^{i}\Vert -\Vert 
\mathbf{x}_{i}\Vert \geq 2R^{\prime }-R^{\prime }=R^{\prime }.
\end{equation*}%
By mean of this inequality and monotonousness of $\rho _{1}$, obtain 
\begin{equation*}
\frac{1}{2}=\frac{1}{n_{j}}\sum_{i=1}^{n_{j}}\rho _{1}\left( \frac{d(\mathbf{%
x}_{i},\mathcal{A}_{n_{j}})}{M(\mathcal{A}_{n_{j}},F_{n_{j}})}\right) \geq 
\frac{1}{n_{j}}\sum_{\mathbf{x}_{i}\in B(R^{\prime })}\rho _{1}\left( \frac{%
R^{\prime }}{M(\mathcal{A}_{n_{j}},F_{n_{j}})}\right) .
\end{equation*}%
The Lemma Hypothesis, enable us to take $j_{0}$ such that $M(\mathcal{A}%
_{n_{j}},F_{n_{j}})<m^{\ast }\ \ \forall j>j_{0}$, then, for $j>j_{0}$ we
get 
\begin{equation*}
\frac{1}{2}\geq \frac{1}{n_{j}}\sum_{\mathbf{x}_{i}\in B(R^{\prime })}\rho
_{1}\left( \frac{R^{\prime }}{M(\mathcal{A}_{n_{j}},F_{n_{j}})}\right) \geq 
\frac{1}{n_{j}}\sum_{\mathbf{x}_{i}\in B(R^{\prime })}\rho _{1}\left( \frac{%
R^{\prime }}{m^{\ast }}\right) .
\end{equation*}%
Applying condition $(a)$ over $R^{\prime }$ we obtain 
\begin{equation*}
\frac{1}{2}\geq \frac{1}{n_{j}}\sum_{\mathbf{x}_{i}\in B(R^{\prime })}\frac{3%
}{4}=\frac{3}{4}\mathbb{P}_{n_{j}}(B(R^{\prime })).
\end{equation*}%
As $\omega $ is such that $\mathbb{P}_{n_{j}}(B(R^{\prime }))\rightarrow 
\mathbb{P}(B(R^{\prime }))$ when $j\rightarrow \infty $, taking limit at the
inequality above we arrive to ${1}/{2}\geq {3}\mathbb{P}(B(R^{\prime }))/{4}$%
, but condition (b) at $R^{\prime }$ implies $\mathbb{P}(B(R^{\prime }))>{2}/%
{3}$, therefore 
\begin{equation*}
\frac{1}{2}\geq \frac{3}{4}\mathbb{P}(B(R^{\prime }))>\frac{3}{4}\frac{2}{3}=%
\frac{1}{2},
\end{equation*}%
which is absurd. 
\end{proof}

The following Lemma expresses that if only centers belonging in a ball are
considered, scale estimator changes just a little provided the ball is big
enough.

\begin{lemma}
\label{lemmaLimiteDoble} Let $\mathcal{A}_{n}$ be a sequence of $k$-points,
such that $M(\mathcal{A}_{n},F_{n})$ is bounded in the sense of Lemma \ref%
{escalaAcotada}. Let $B(R)\subseteq \mathbb{R}^{p}$ be the ball centered at $%
0$ with radius $R$. It define $\mathcal{A}_{n}^{R}=\mathcal{A}_{n}\cap B(R)$%
. Besides, suppose $(\mathbf{x}_{n})_{n\in \mathbb{N}}$ are $i.i.d^{\prime
}s $. Then, 
\begin{equation}
\forall \varepsilon >0\quad \exists R\quad \mathbb{P}\left( \{\omega
:\limsup_{n}\ M(\mathcal{A}_{n}^{R},F_{n})-M(\mathcal{A}_{n},F_{n})<\epsilon
\}\right) =1  \label{LimiteDobleM}
\end{equation}

and 
\begin{equation}
\forall \varepsilon >0\quad \exists R\quad \mathbb{P}\left( \{\omega
:\limsup_{n}\ \tau (\mathcal{A}_{n}^{R},F_{n})-\tau (\mathcal{A}%
_{n},F_{n})<\epsilon \}\right) =1.  \label{LimiteDobleTau}
\end{equation}
\end{lemma}

The proof of this Lemma will be done first for the M-scale, and second for
the $\tau $-scale. 

\begin{proof}[ Proof of Lemma \protect\ref{lemmaLimiteDoble} for the M-scale]

From Hypothesis {A.1)}, by a compactness argument, we can obtain a positive
number $\delta ^{\ast }$ such that every collections of $k$ balls of radius $%
\delta ^{\ast }$, has probability less than $0.5$. More precisely, there
exists $\delta ^{\ast }>0$ such that if $\lambda $ is 
\begin{equation}
\lambda =\sup \ \left\{ \mathbb{P}(\mathbf{x}\in \cup _{\boldsymbol{\mu }\in 
\mathcal{A}}B(\boldsymbol{\mu },\delta ^{\ast })):\mathcal{A}\subseteq 
\mathbb{R}^{p},\#(\mathcal{A})\leq k\right\} ,  \label{seDesprendedeH1a}
\end{equation}%
then $\lambda <1/2.$ In turn, by using the identity of the following sets 
\begin{equation*}
\cap _{\boldsymbol{\mu }\in \mathcal{A}}B^{c}(\boldsymbol{\mu },\delta
^{\ast })=\{\mathbf{x}:d(\mathbf{x},\mathcal{A})\geq \delta ^{\ast }\},
\end{equation*}%
and taking probability (\ref{seDesprendedeH1a}) its equivalent to 
\begin{equation}
\frac{1}{2}<a=\inf \ \left\{ \mathbb{P}(\{\mathbf{x}:d(\mathbf{x},\mathcal{A}%
)\geq \delta ^{\ast }\}):\mathcal{A}\subseteq \mathbb{R}^{p},\#(\mathcal{A}%
)\leq k\right\} ,  \label{seDesprendedeH1}
\end{equation}
where $a=1-\lambda >1/2\ $.

We will demonstrate equation (\ref{LimiteDobleM}) of Lemma. Suppose the
opposite holds, then there exists $\beta >0$, such that the set 
\begin{equation}
\Omega ^{\prime }(R_{0})=\left\{ \omega \in \Omega :\limsup_{n}\ M({\mathcal{%
A}_{n}^{(\omega )}}^{R_{0}},F_{n}^{(\omega )})-M(\mathcal{A}_{n}^{(\omega
)},F_{n}^{(\omega )})\geq \beta \right\}  \label{omegaR}
\end{equation}%
has positive probability for all $R_{0}.$ On the other hand, consider $%
\delta ^{\ast }$ as equation (\ref{seDesprendedeH1}), then if 
\begin{equation}
\Omega _{2}^{\prime }=\left\{ \omega \in \Omega :\lim_{n\rightarrow \infty
}\inf_{\substack{ \mathcal{A}\subseteq \mathbb{R}^{p}  \\ \#(\mathcal{A}%
)\leq k}}\frac{1}{n}\#\left( \{i:d(\mathbf{x}_{i}^{(\omega )},\mathcal{A}%
)\geq \delta ^{\ast }\}\right) =\inf_{\substack{ \mathcal{A}\subseteq 
\mathbb{R}^{p},  \\ \#(\mathcal{A})\leq k}}\mathbb{P}(d(\mathbf{x},\mathcal{A%
})\geq \delta ^{\ast })\right\} ,  \label{Omega2}
\end{equation}%
by Lemma \ref{LemaPollardhcero}, $\mathbb{P}(\Omega _{2}^{\prime })=1$. This
will be used later.

To demonstrate the Lemma we will set the constants as follows. Let be $a>{1}/%
{2}$ from equation (\ref{seDesprendedeH1}), from limit inequality 
\begin{equation*}
\lim_{\kappa \rightarrow 0^{+}}(1-\kappa )(a(1-\kappa )-\kappa )=a>\frac{1}{2%
},
\end{equation*}%
we can choose $\kappa _{0}\in (0,1)$ such that 
\begin{equation}
(1-\kappa _{0})(a(1-\kappa _{0})-\kappa _{0})>\frac{1}{2}.
\label{eleccionKappa}
\end{equation}%
On the other hand, let $m^{\ast }$ be the bound for $\lim \sup
M(A_{n},F_{n})<m^{\ast }$. We apply Lemma \ref{propiedadderho} with
constants $t$ and $\alpha $ 
\begin{equation}
t=\frac{\beta \delta ^{\ast }}{m^{\ast }(m^{\ast }+\beta )}>0,\quad \alpha =%
\frac{\delta ^{\ast }}{m^{\ast }+\beta }.  \label{elecciontalpha}
\end{equation}%
Let $\kappa _{0}$ be defined in (\ref{eleccionKappa}), take $\gamma
_{0}:=\gamma (\kappa _{0})>0$ for which Lemma \ref{propiedadderho} holds,
that is 
\begin{equation}
\forall u\geq \alpha ,\Delta \geq t,\quad \mbox{tales que}\quad |\rho
(u+\Delta )-\rho (u)|<\gamma _{0}\quad \Rightarrow \quad \rho (u)>1-\kappa
_{0}.  \label{elecciongamma}
\end{equation}%
In turn, let $R$ be a real number big enough such that 
\begin{equation}
\mathbb{P}(B(R/4))>1-\kappa _{0}\gamma _{0}\quad y\quad R>2R_{1},
\label{eleccionConstR}
\end{equation}%
where $R_{1}$ is obtained from applying el Lemma \ref{escalaAcotada} for
this case. So we determine the subsets needed for the proof as follows, to $%
R $ defined in (\ref{eleccionConstR}), we take $R_{0}=R$ in the equation (%
\ref{omegaR}), so that $\Omega _{1}^{\prime }:=\Omega ^{\prime }(R)$. Let $%
\Omega _{3}^{\prime }$ be the set from Lemma \ref{escalaAcotada}, that is,
for all $\omega $ belonging to $\Omega _{3}^{\prime }$ there exists $%
n_{0}(\omega )$ such that for all $n\geq n_{0}$ is simultaneously valid: 
\begin{equation}
(i)M(\mathcal{A}_{n},F_{n})\leq m^{\ast }\quad \quad (ii)\mathcal{A}_{n}\cap
B(R)\neq \emptyset ,  \label{Omega3}
\end{equation}%
and let $\Omega _{4}^{\prime }$ be the set where $\mathbb{P}%
_{n}(B(R/4))\rightarrow \mathbb{P}(B(R/4))$. $\Omega _{i}^{\prime }$ with $%
i=2,3,4$ has probability $1$. Whereas we have assumed $\mathbb{P}(\Omega
_{1}^{\prime })>0$, then $\mathbb{P}(\Omega _{1}^{\prime }\cap \Omega
_{2}^{\prime }\cap \Omega _{3}^{\prime }\cap \Omega _{4}^{\prime })>0$.
Hence, it is possible to take $\omega $ in the intersection, that will keep
fixed throughout the development of the proof and with which we will arrive
to an absurd. For this $\omega $, 
\begin{equation*}
\limsup_{n}\ M({\mathcal{A}_{n}^{(\omega )}}^{R},F_{n}^{(\omega )})-M(%
\mathcal{A}_{n}^{(\omega )},F_{n}^{(\omega )})\geq \beta ,
\end{equation*}%
then, by considering subsequences, we can suppose that there exists an
infinite set $\mathbb{N}^{\prime }\subseteq \mathbb{N}$ such that 
\begin{equation*}
M(\mathcal{A}_{n}^{R},F_{n})-M(\mathcal{A}_{n},F_{n})\geq \beta ,\quad
\forall n\in \mathbb{N}^{\prime }
\end{equation*}%
operating, we get 
\begin{equation*}
\rho _{1}\left( \frac{d(\mathbf{x}_{i},\mathcal{A}_{n}^{R})}{M(\mathcal{A}%
_{n}^{R},F_{n})}\right) \leq \rho _{1}\left( \frac{d(\mathbf{x}_{i},\mathcal{%
A}_{n}^{R})}{M(\mathcal{A}_{n},F_{n})+\beta }\right)
\end{equation*}%
that implies, 
\begin{equation}
\frac{1}{n}\sum_{i=1}^{n}\rho _{1}\left( \frac{d(\mathbf{x}_{i},\mathcal{A}%
_{n}^{R})}{M(\mathcal{A}_{n}^{R},F_{n})}\right) \leq \frac{1}{n}%
\sum_{i=1}^{n}\rho _{1}\left( \frac{d(\mathbf{x}_{i},\mathcal{A}_{n}^{R})}{M(%
\mathcal{A}_{n},F_{n})+\beta }\right) .  \label{eqpivote}
\end{equation}%
As $\omega \in \Omega _{3},\exists \ n_{0}$ such that $\forall n>n_{0},n\in 
\mathbb{N}^{\prime }$, there are always be at least an element, say $%
\boldsymbol{\mu }_{1}^{n}$, in the ball of radius $R_{1}$, where $R_{1}$ is
given by Lemma \ref{escalaAcotada}. The choice of $R$ made in (\ref%
{eleccionConstR}) implies that $R_{1}<R/2$. Then, $\mathbf{x}\in B(R/4)$
verifies:%
\begin{equation*}
d(\mathbf{x},\mathcal{A}_{n})\leq \Vert \mathbf{x}-\boldsymbol{\mu }%
_{1}^{n}\Vert \leq \Vert \mathbf{x}\Vert +\Vert \boldsymbol{\mu }%
_{1}^{n}\Vert \leq R/4+R_{1}<R/4+R/2=3R/4.
\end{equation*}%
Besides, if $\boldsymbol{\mu }_{a}^{n}\in \mathcal{A}_{n}\setminus \mathcal{A%
}_{n}^{R}$, then $\Vert \mathbf{x}-\boldsymbol{\mu }_{a}^{n}\Vert \geq \Vert 
\boldsymbol{\mu }_{a}^{n}\Vert -\Vert \mathbf{x}\Vert \geq 3R/4$, therefore, 
$\mathbf{x}$ always be nearer from $\boldsymbol{\mu }_{1}^{n}$ than from any
other center outside the ball $B(R)$, and then $d(\mathbf{x},\mathcal{A}%
_{n})=d(\mathbf{x},\mathcal{A}_{n}^{R})$. So, we can rewrite the right hand
side of equation (\ref{eqpivote}), by dividing it into two summations
regarding the belonging of $\mathbf{x}_{i}$ to $B(R/4)$:%
\begin{align}
& \frac{1}{n}\sum_{\mathbf{x}_{i}\in B(R/4)}\rho _{1}\left( \frac{d(\mathbf{x%
}_{i},\mathcal{A}_{n}^{R})}{M(\mathcal{A}_{n},F_{n})+\beta }\right) & +& 
\frac{1}{n}\sum_{\mathbf{x}_{i}\notin B(R/4)}\rho _{1}\left( \frac{d(\mathbf{%
x}_{i},\mathcal{A}_{n}^{R})}{M(\mathcal{A}_{n},F_{n})+\beta }\right) \leq 
\notag \\
{}\leq & \frac{1}{n}\sum_{\mathbf{x}_{i}\in B(R/4)}\rho _{1}\left( \frac{d(%
\mathbf{x}_{i},{\mathcal{A}_{n}})}{M(\mathcal{A}_{n},F_{n})+\beta }\right) & 
+& \frac{\#(\{\mathbf{x}_{i}\notin B(R/4)\})}{n}\leq  \notag \\
{}\leq & \frac{1}{n}\sum_{{i=1}}^{{n}}\rho _{1}\left( \frac{d(\mathbf{x}_{i},%
\mathcal{A}_{n})}{M(\mathcal{A}_{n},F_{n})+\beta }\right) & +& \frac{\#(\{%
\mathbf{x}_{i}\notin B(R/4)\})}{n},  \label{quitarUnaR}
\end{align}%
combining equations (\ref{eqpivote}) and (\ref{quitarUnaR}), 
\begin{equation*}
\frac{1}{n}\sum_{i=1}^{n}\rho _{1}\left( \frac{d(\mathbf{x}_{i},\mathcal{A}%
_{n}^{R})}{M(\mathcal{A}_{n}^{R},F_{n})}\right) \leq \frac{1}{n}\sum_{{i=1}%
}^{{n}}\rho _{1}\left( \frac{d(\mathbf{x}_{i},\mathcal{A}_{n})}{M(\mathcal{A}%
_{n},F_{n})+\beta }\right) +\frac{\#(\{\mathbf{x}_{i}\notin B(R/4)\})}{n},
\end{equation*}%
denoting 
\begin{equation*}
\mathbb{P}_{n}(B(R/4)^{c})=\frac{\#(\{\mathbf{x}_{i}\notin B(R/4)\})}{n},
\end{equation*}%
Making a passage of terms to the left in the previous equation we obtain 
\begin{equation}
0\leq \frac{1}{n}\sum_{i=1}^{n}\rho _{1}\left( \frac{d(\mathbf{x}_{i},%
\mathcal{A}_{n})}{M(\mathcal{A}_{n},F_{n})}\right) -\rho _{1}\left( \frac{d(%
\mathbf{x}_{i},\mathcal{A}_{n})}{M(\mathcal{A}_{n},F_{n})+\beta }\right)
\leq \mathbb{P}_{n}(B(R/4)^{c}).  \label{eqdesigualdadPn}
\end{equation}%
Consider the sets 
\begin{equation*}
\mathcal{C}_{n}=\{\mathbf{x}_{i}:d(\mathbf{x}_{i},\mathcal{A}_{n})\geq
\delta ^{\ast }\}\ 
\end{equation*}%
and%
\begin{equation*}
\mathcal{D}_{n}(\gamma _{0})=\left\{ \mathbf{x}_{i}:\rho _{1}\left( \frac{d(%
\mathbf{x}_{i},\mathcal{A}_{n})}{M(\mathcal{A}_{n},F_{n})}\right) -\rho
_{1}\left( \frac{d(\mathbf{x}_{i},\mathcal{A}_{n})}{M(\mathcal{A}%
_{n},F_{n})+\beta }\right) \geq \gamma _{0}\right\} ,
\end{equation*}%
where $\delta ^{\ast }$ is the constant defined in (\ref{seDesprendedeH1})
and $\gamma _{0}>0$ is the choice corresponding to (\ref{elecciongamma}).
From equation (\ref{eqdesigualdadPn}) it can be deduced straightforwardly
that 
\begin{equation}
\mathbb{P}_{n}(\mathcal{D}_{n}(\gamma _{0}))\leq \frac{1}{\gamma _{0}}%
\mathbb{P}_{n}(B(R/4)^{c}).  \label{cotaParaDn}
\end{equation}%
On the other hand, values that were taken in (\ref{eleccionKappa}) and (\ref%
{elecciontalpha}) for determining $\gamma _{0}$, allow us to apply the Lemma %
\ref{propiedadderho}, and in this way we get 
\begin{equation}
\mathbf{x}_{i}\in \mathcal{C}_{n}\cap \mathcal{D}_{n}(\gamma _{0})^{c},\quad %
\mbox{entonces }\quad \rho _{1}\left( \frac{d(\mathbf{x}_{i},\mathcal{A}_{n})%
}{M(\mathcal{A}_{n},F_{n})}\right) >1-\kappa _{0}.  \label{costoUnHuevo}
\end{equation}%
Indeed, defining 
\begin{equation*}
u_{i}=\frac{d(\mathbf{x}_{i},\mathcal{A}_{n})}{M(\mathcal{A}%
_{n},F_{n})+\beta },
\end{equation*}%
thus, ${D}_{n}(\gamma _{0})^{c}$ can be represented as 
\begin{equation*}
{D}_{n}(\gamma _{0})^{c}=\left\{ \mathbf{x}_{i}:\rho _{1}(u_{i}+\Delta
_{i})-\rho (u_{i})\leq \gamma _{0}\right\} ,
\end{equation*}%
where 
\begin{equation*}
\Delta _{i}:=\frac{d(\mathbf{x}_{i},\mathcal{A}_{n})}{M(\mathcal{A}%
_{n},F_{n})}-\frac{d(\mathbf{x}_{i},\mathcal{A}_{n})}{M(\mathcal{A}%
_{n},F_{n})+\beta }
\end{equation*}%
To see that, it is possible to apply Lemma \ref{propiedadderho}, but first,
we must establish conditions that ensure $u_{i}\geq \alpha $ and $\Delta
_{i}\geq t$, where $\alpha $ and $t$ are defined in (\ref{elecciontalpha}).
Namely, for condition about $n$: $\omega \in \Omega _{3}$, then we have that 
$\lim \sup M(\mathcal{A}_{n},F_{n})<m^{\ast }$. So we can take $n_{1}>n_{0}$%
, from which $M(\mathcal{A}_{n},F_{n})\leq m^{\ast }$. On the other hand,
let $\mathbf{x}_{i}$ be a point from $\mathcal{C}_{n}$, then $d(\mathbf{x}%
_{i},\mathcal{A}_{n})\geq \delta ^{\ast }$. It is easy to see that 
\begin{equation*}
u_{i}\geq \frac{\delta ^{\ast }}{(m^{\ast }+\beta )}=\alpha .
\end{equation*}%
By last, condition on $\Delta _{i}$ follows from a straight forward
computation, 
\begin{equation*}
\Delta _{i}\geq \frac{\beta \delta ^{\ast }}{m^{\ast }(m^{\ast }+\beta )}=t.
\end{equation*}%
In this way, we apply the Lemma \ref{propiedadderho} for each $u_{i}$ under
conditions: $\mathbf{x}_{i}\in \mathcal{C}_{n}\cap {D}_{n}(\gamma _{0})^{c}$%
, $n\in \mathbb{N}^{\prime },$ and $n>n_{1}$. We obtain $\rho
_{1}(u_{i})>1-\kappa _{0}$, that means, 
\begin{equation*}
\rho _{1}\left( \frac{d(\mathbf{x}_{i},\mathcal{A}_{n})}{M(\mathcal{A}%
_{n},F_{n})+\beta }\right) >1-\kappa _{0}.
\end{equation*}%
Finally, given the monotonicity of $\rho _{1}$, we have that 
\begin{equation*}
\rho _{1}\left( \frac{d(\mathbf{x}_{i},\mathcal{A}_{n})}{M(\mathcal{A}%
_{n},F_{n})}\right) >1-\kappa _{0}
\end{equation*}%
as we wanted to demonstrate. Hence, implication (\ref{costoUnHuevo}) has
been proved, and we use it in the following inequality 
\begin{equation*}
\frac{1}{2}=\frac{1}{n}\sum_{i=1}^{n}\rho _{1}\left( \frac{d(\mathbf{x}_{i},%
\mathcal{A}_{n})}{M(\mathcal{A}_{n},F_{n})}\right) \geq \frac{1}{n}\sum_{%
\mathbf{x}_{i}\in \mathcal{C}_{n}\cap \mathcal{D}_{n}(\gamma _{0})^{c}}\rho
_{1}\left( \frac{d(\mathbf{x}_{i},\mathcal{A}_{n})}{M(\mathcal{A}_{n},F_{n})}%
\right) \geq
\end{equation*}%
\begin{equation}
\geq \frac{1}{n}\sum_{\mathbf{x}_{i}\in \mathcal{C}_{n}\cap \mathcal{D}%
_{n}(\gamma _{0})^{c}}1-\kappa _{0}=(1-\kappa _{0})\mathbb{P}_{n}(\mathcal{C}%
_{n}\cap \mathcal{D}_{n}(\gamma _{0})^{c}).  \label{23}
\end{equation}%
Using $\mathbb{P}_{n}(\mathcal{C}_{n}\cap \mathcal{D}_{n}(\gamma
_{0})^{c})\geq \mathbb{P}_{n}(\mathcal{C}_{n})+\mathbb{P}_{n}(\mathcal{D}%
_{n}(\gamma _{0})^{c})-1$ in equation (\ref{23}), we obtain 
\begin{equation*}
\frac{1}{2}\geq (1-\kappa _{0})\left( \mathbb{P}_{n}(\mathcal{C}_{n})+%
\mathbb{P}_{n}(\mathcal{D}_{n}(\gamma _{0})^{c})-1\right) =(1-\kappa
_{0})\left( \mathbb{P}_{n}(\mathcal{C}_{n})-\mathbb{P}_{n}(\mathcal{D}%
_{n}(\gamma _{0}))\right) .
\end{equation*}%
From previous equation and equation (\ref{cotaParaDn}) we derive 
\begin{equation}
\frac{1}{2}\geq (1-\kappa _{0})(\mathbb{P}_{n}(\mathcal{C}_{n})-\frac{1}{%
\gamma _{0}}\mathbb{P}_{n}(B(R/4)^{c})).  \label{ultimaDesigualdad}
\end{equation}%
As $\omega \in \Omega _{2}^{\prime }\cap \Omega _{4}^{\prime }$, we can
choose, $n_{2}\in \mathbb{N}^{\prime },n_{2}>n_{1}$ such that 
\begin{equation}
\mathbb{P}_{n_{2}}(\mathcal{C}_{n_{2}})\geq \inf_{\mathcal{A}}\frac{1}{n_{2}}%
\#\left( \{\mathbf{x}_{i}:d(\mathbf{x}_{i},\mathcal{A})\geq \delta ^{\ast
}\}\right) \geq a(1-\kappa _{0}),  \label{bul1}
\end{equation}%
where $\kappa _{0}\in (0,1)$ is defined in (\ref{eleccionKappa}). Indeed,
this is because that the empirical probability 
\begin{equation*}
\inf_{\mathcal{A}}\frac{1}{n}\#\left( \{\mathbf{x}_{i}:d(\mathbf{x}_{i},%
\mathcal{A})\geq \delta ^{\ast }\}\right)
\end{equation*}%
converges to the probability of it population version, that is, 
\begin{equation}
\lim_{n\rightarrow \infty }\inf_{\mathcal{A}}\frac{1}{n}\#\left( \{\mathbf{x}%
_{i}:d(\mathbf{x}_{i},\mathcal{A})\geq \delta ^{\ast }\}\right) =\inf_{%
\mathcal{A}}\mathbb{P}\left( \{\mathbf{x}:d(\mathbf{x}_{i},\mathcal{A})\geq
\delta ^{\ast }\}\right) =a>a(1-\kappa _{0}),  \label{whatweneed}
\end{equation}%
where $a$ is defined in (\ref{seDesprendedeH1}).

From (\ref{whatweneed}), we obtain the existence of some $n_{2}$ fulfilling (%
\ref{bul1}). On the other hand, let's see that 
\begin{equation}
\mathbb{P}_{n_{2}}(B(R/4)^{c})<\gamma _{0}\kappa _{0},  \label{bul2}
\end{equation}%
since $\omega \in \Omega _{4}^{\prime }$, the set where $\mathbb{P}%
_{n}(B(R/4))\rightarrow \mathbb{P}(B(R/4)),$ and in turn as is deducted from
(\ref{eleccionConstR}), $\mathbb{P}(B(R/4)^{c})<\gamma _{0}\kappa _{0}$. 

Hence, applying inequalities (\ref{bul1}) and (\ref{bul2}) on (\ref%
{ultimaDesigualdad}), we have that 
\begin{equation*}
\frac{1}{2}\geq (1-\kappa _{0})(a(1-\kappa _{0})-\kappa _{0}),
\end{equation*}%
which contradicts the choice of $\kappa _{0}$ made in (\ref{eleccionKappa}).
This concludes the proof, because we arrive to an contradiction from
supposing that there are $R$ and $\varepsilon $ for which the thesis that
Lemma does not occur.
\end{proof}

\bigskip

\begin{proof}[Proof of Lemma \protect\ref{lemmaLimiteDoble} for the $\protect%
\tau $ -scale]
\label{pruebaLemaR} First of all, we will use the mean value Theorem for the
function $f(s)=s^{2}\rho _{2}(\frac{r}{s})$: 
\begin{equation}
\left\vert s_{2}^{2}\rho _{2}\left( \frac{r}{s_{2}}\right) -s_{1}^{2}\rho
_{2}\left( \frac{r}{s_{1}}\right) \right\vert \leq C_{0}\xi |s_{2}-s_{1}|,
\label{TVM}
\end{equation}%
where $C_{0}=\sup_{u\in \mathbb{R}}{\ |2\rho _{2}(u)-\psi _{2}(u)u|}$, and $%
\xi \in (s_{1},s_{2})$. Let $m^{\ast }=\lim \sup M(A_{n},F_{n})$. Let $%
\varepsilon >0$ be a small number, we take $R$ satisfying simultaneously, 
\begin{equation}
\mathbb{P}(B(R/4))>1-\frac{\varepsilon }{2}\quad ,\quad R>2R_{1},
\label{eleccionConstRTau}
\end{equation}%
where $R_{1}$ is defined in equation (\ref{eleccionConstR}). Besides, by
applying (\ref{LimiteDobleM}), we can choose $R$ that also satisfies 
\begin{equation}
\limsup_{n}\ M(\mathcal{A}_{n}^{R},F_{n})-M(\mathcal{A}_{n},F_{n})<\frac{%
\varepsilon }{2C_{0}(1+m^{\ast })}.  \label{eleccionConstRTau2}
\end{equation}%
Now, we demonstrate equation (\ref{LimiteDobleTau}) from Lemma, for this, 
\begin{equation*}
\tau ^{2}(\mathcal{A}_{n}^{R},F_{n})-\tau ^{2}(\mathcal{A}_{n},F_{n})=\frac{1%
}{n}\sum_{i=1}^{n}M(\mathcal{A}_{n}^{R},F_{n})^{2}\rho _{2}\left( \frac{d(%
\mathbf{x}_{i},\mathcal{A}_{n}^{R})}{M(\mathcal{A}_{n}^{R},F_{n})}\right) -M(%
\mathcal{A}_{n},F_{n})^{2}\rho _{2}\left( \frac{d(\mathbf{x}_{i},\mathcal{A}%
_{n})}{M(\mathcal{A}_{n},F_{n})}\right) .
\end{equation*}%
Reasoning analogously to equation (\ref{quitarUnaR}), the previous equation
becomes 
\begin{equation*}
\leq \frac{1}{n}\sum_{i=1}^{n}M(\mathcal{A}_{n}^{R},F_{n})^{2}\rho
_{2}\left( \frac{d(\mathbf{x}_{i},\mathcal{A}_{n})}{M(\mathcal{A}%
_{n}^{R},F_{n})}\right) -M(\mathcal{A}_{n},F_{n})^{2}\rho _{2}\left( \frac{d(%
\mathbf{x}_{i},\mathcal{A}_{n})}{M(\mathcal{A}_{n},F_{n})}\right) +\frac{%
\#(\{\mathbf{x}_{i}\notin B(R/4)\})}{n},
\end{equation*}%
by applying results given in (\ref{TVM}), we obtain 
\begin{equation*}
\leq \frac{1}{n}\sum_{i=1}^{n}C_{0}\xi _{i}|M(\mathcal{A}_{n}^{R},F_{n})-M(%
\mathcal{A}_{n},F_{n})|+\frac{\#(\{\mathbf{x}_{i}\notin B(R/4)\})}{n},
\end{equation*}%
where $\xi _{i}\in (M(\mathcal{A}_{n},F_{n}),M(\mathcal{A}_{n}^{R},F_{n}))$.
Then, $\xi _{i}\leq M(\mathcal{A}_{n}^{R},F_{n})$, and $M(\mathcal{A}%
_{n}^{R},F_{n})$ is bounded by $1+m^{\ast }$ since (\ref{LimiteDobleM}).
Thus, 
\begin{equation*}
\tau ^{2}(\mathcal{A}_{n}^{R},F_{n})-\tau ^{2}(\mathcal{A}_{n},F_{n})\leq
C_{0}(1+m^{\ast })|M(\mathcal{A}_{n}^{R},F_{n})-M(\mathcal{A}_{n},F_{n})|+%
\mathbb{P}_{n}(B(R/4)).
\end{equation*}%
Taking limsup at the right hand side, and using condition for $R$ in (\ref%
{eleccionConstRTau2}) and (\ref{eleccionConstRTau}) we get 
\begin{equation*}
\lim \sup C_{0}(1+m^{\ast })|M(\mathcal{A}_{n}^{R},F_{n})-M(\mathcal{A}%
_{n},F_{n})|+\mathbb{P}_{n}(B(R/4))<C_{0}(1+m^{\ast })\frac{\varepsilon }{%
2C_{0}(1+m^{\ast })}+\frac{\varepsilon }{2}=\varepsilon ,
\end{equation*}%
and therefore, there exists $n_{0}$ such that 
\begin{equation*}
\tau ^{2}(\mathcal{A}_{n}^{R},F_{n})-\tau ^{2}(\mathcal{A}_{n},F_{n})\leq
\varepsilon ,
\end{equation*}%
from this, the conclusion of Lemma can be easily obtained.
\end{proof}

\begin{lemma}
\label{lemaCompacto1} Let $(\mathbf{x}_{i})_{i\in \mathbb{N}}$ be i.i.d's,
defined on a probability space $(\Omega ,\mathcal{F},\mathbb{P})$. Let $%
\mathcal{A}_{n}$ be the corresponding optimal $\tau $-centers with a number
of centers less than or equal $K$ based on $\{\mathbf{x}_{1},\dots ,\mathbf{x%
}_{n}\}$. Under the general Hypothesis A.1) and A.2), exists a constant $C>0$
and a set $\Omega ^{\prime }\subseteq \Omega ,\mathbb{P}(\Omega ^{\prime
})=1 $, satisfying

\begin{equation}
\limsup_{n} d_H( \mathcal{A}_{n}^{(\omega)},\mathcal{A}_{0}) \leq C \ \ \
\forall \omega \in \Omega^{\prime }  \label{centroidesCompacto}
\end{equation}
\end{lemma}

\begin{proof}

\textit{Lemma \ref{lemaCompacto1} will be proved by absurd, reasoning as
follows, if clusters centers are not into a compact set then we could
suppose that there will be at least a center which is located in a very far
region with small probability. Then, considering the optimal cluster without
its furthest center will have a negligible impact in the $tau$ scale, that
would mean that the optimal scale value obtained with $K$ centers can be
compared to the optimal with $K-1 $ centers. But, that would be a
contradiction, because the $K$ optimal center it is well separated from the
optimal with $K-1$.}

Define $\nu_{0}^{(k-1)}$ to infimum of $\tau$ scale of $k-1$ centers, that
is, 
\begin{equation*}
\nu _{0}=\inf_{\#\mathcal{A}\leq k}\tau (\mathcal{A},F),\quad \quad \nu
_{0}^{(k-1)}=\inf_{\#\mathcal{A}\leq k-1}\tau (\mathcal{A},F).
\end{equation*}

Due to {A.2} $\nu _{0}^{(k-1)}-\nu _{0}>0$ and take $\varepsilon=(\nu
_{0}^{(k-1)}-\nu _{0})/2>0$. Choose the subset $\Omega _{1}^{\prime } $ from
Lemma \ref{lemmaLimiteDoble}, whit probability $1$ and its corresponding $%
R_{0}$ . Thus $\Omega _{1}^{\prime }$ has the following property: $\forall
\omega \in \Omega _{1}^{\prime }$ $\exists n_{0}(\omega )$ such that 
\begin{equation}
\tau ({\mathcal{A}_{n}^{R_{0}}},F_{n})-\tau (\mathcal{A}_{n},F_{n})<%
\varepsilon \ \ \ \forall n\geq n_{0}.  \label{propiedadOmega}
\end{equation}%
For that $R_{0}$ (independently of $\omega $), take the compact set $%
\boldsymbol{\mathcal{K}}_{1}$ 
\begin{equation}
\boldsymbol{\mathcal{K}}_{1}=\left\{ \mathcal{A}:\mathcal{A}\subseteq 
\overline{B(R_{0})},\#\left( \mathcal{A}\right) \leq k-1\right\} .
\label{compactoK1}
\end{equation}%
Then, considering the set of probability $1$ $\Omega _{2}^{\prime }$, where
Lemma \ref{ConvUnifacop} occurs for $\boldsymbol{\mathcal{K}}_{1}$. That is, 
$\forall \omega \in \Omega _{2}^{\prime }$, 
\begin{equation*}
\lim_{n\rightarrow \infty }\sup_{\mathcal{A}\in \boldsymbol{\mathcal{K}}%
_{1}}\left\vert \tau (\mathcal{A},F)-\tau (\mathcal{A},F_{n}^{(\omega
)})\right\vert =0.
\end{equation*}%
Finally, let $\Omega _{3}^{\prime }$ be the set where the following limit
happens 
\begin{equation*}
\tau (\mathcal{A}_{0},F_{n})\rightarrow \tau (\mathcal{A}_{0},F)=\nu _{0},
\end{equation*}%
We will prove that $\mathbb{P}(\Omega _{3}^{\prime })=1$. We use an absurd
reasoning , if we suppose the opposite to the we want to prove ( denying
equation (\ref{centroidesCompacto})), we obtain 
\begin{equation*}
\forall C>0\quad \mathbb{P}\left( \left\{ \omega \in \Omega
:\limsup_{n}d_{H}(\mathcal{A}_{n}^{(\omega )},\mathcal{A}_{0})>C\right\}
\right) \neq 0
\end{equation*}%
Define $C_{0}:=R_{0}+\max_{\boldsymbol{\mu }\in \mathcal{A}_{0}}\boldsymbol{%
\mu }$, and consider the set whit positive probability $\Omega
_{C_{0}}=\{\omega \in \Omega :\limsup_{n}d_{H}(\mathcal{A}_{n}^{(\omega )},%
\mathcal{A}_{0})>C_{0}\}$. By guessing $\mathbb{P}(\Omega _{C_{0}})>0$, it
is possible to take $\omega \in \Omega _{C_{0}}\cap \Omega _{1}^{\prime
}\cap \Omega _{2}^{\prime }\cap \Omega _{3}^{\prime }$, We will keep fixed
this $\omega $ thorough the proof and arrive to a contradiction. Let $(%
\mathcal{A}_{n})_{n\in \mathbb{N}_{1}}$ be a subsequence such that $d_{H}(%
\mathcal{A}_{n}^{(\omega )},\mathcal{A}_{0})>C_{0}$, then, as Hausdorff
distance for finite set is always achieved, there exists, $\widetilde{%
\boldsymbol{\mu }}_{n}\in \mathcal{A}_{n}^{(\omega )}$ and $\widetilde{%
\boldsymbol{\mu }}(n)\in \mathcal{A}_{0}$, such that 
\begin{equation*}
\Vert \widetilde{\boldsymbol{\mu }}_{n}-\widetilde{\boldsymbol{\mu }}%
(n)\Vert =d_{H}(\mathcal{A}_{n}^{(\omega )},\mathcal{A}_{0})>C_{0}.
\end{equation*}%
As $\Vert \widetilde{\boldsymbol{\mu }}_{n}\Vert +\Vert \boldsymbol{\mu }%
(n)\Vert \geq \Vert \widetilde{\boldsymbol{\mu }}_{n}-\widetilde{\boldsymbol{%
\mu }}(n)\Vert $, then 
\begin{equation*}
\Vert \widetilde{\boldsymbol{\mu }}_{n}\Vert >C_{0}-\Vert \widetilde{%
\boldsymbol{\mu }}(n)\Vert \geq C_{0}-\min_{\boldsymbol{\mu }\in \mathcal{A}%
_{0}}\Vert \boldsymbol{\mu }\Vert =R_{0}.
\end{equation*}

Therefore, $\mathcal{A}_{n}$ always has a center outside of $B(R_{0})$, then 
$\mathcal{A}_{n}^{R_{0}}=\mathcal{A}_{n}\cap B(R_{0})$ has at most $K-1$
centers, we also note that $\mathcal{A}_{n}^{R_{0}}$ are in the compact $%
\boldsymbol{\mathcal{K}}_{1}$ defined in (\ref{compactoK1}), as $\boldsymbol{%
\mathcal{K}}_{1}$ is a compact set regarding to Hausdorff distance, by
taking subsequences, we can suppose that there exists $\mathcal{A}^{\ast
}\in \boldsymbol{\mathcal{K}}_{1}$ of at most $K-1$ elements, such that $(%
\mathcal{A}_{n}^{R_{0}})_{n\in \mathbb{N}_{2}}$ converge to $\mathcal{A}%
^{\ast }$ in Hausdorff distance. By using Lemma \ref{ConvUnifacop}, and
continuity of $\tau (\mathcal{A},F)$ regarding its first argument, it is
easy to see that 
\begin{equation}
\tau (\mathcal{A}_{n}^{R_{0}},F_{n})\rightarrow \tau (\mathcal{A}^{\ast
},F)\geq \nu _{0}^{(k-1)}.  \label{absurdo1}
\end{equation}%
Indeed, 
\begin{equation*}
|\tau (\mathcal{A}_{n}^{R_{0}},F_{n})-\tau (\mathcal{A}^{\ast },F)|\leq
|\tau (\mathcal{A}_{n}^{R_{0}},F_{n})-\tau (\mathcal{A}_{n}^{R_{0}},F)|+|%
\tau (\mathcal{A}_{n}^{R_{0}},F)-\tau (\mathcal{A}^{\ast },F)|
\end{equation*}%
\begin{equation*}
\leq \sup_{\mathcal{A}\in \boldsymbol{\mathcal{K}}_{1}}|\tau (\mathcal{A}%
,F_{n})-\tau (\mathcal{A},F)|+|\tau (\mathcal{A}_{n}^{R_{0}},F)-\tau (%
\mathcal{A}^{\ast },F)|,
\end{equation*}%
and taking limit in $n$ at the previous inequality, (\ref{absurdo1}) holds.
Also, given that $\mathcal{A}_{n}$ are sample optimum, results 
\begin{equation}
\tau (\mathcal{A}_{n},F_{n})\leq \tau (\mathcal{A}_{0},F_{n})\rightarrow \nu
_{0}.  \label{absurdo2}
\end{equation}
We can find $n_{1}\in \mathbb{N}_{2}$ satisfying simultaneously:

\begin{description}
\item (a) $\tau (\mathcal{A}_{n_{1}}^{R_{0}},F_{n_{1}})\geq \nu
_{0}^{(k-1)}-\varepsilon /4$, this is possible due to equation (\ref%
{absurdo1}) 

\item (b) $\tau (\mathcal{A}_{n_{1}},F_{n_{1}})\leq \nu _{0}+\varepsilon /4$%
, this is possible due to equation (\ref{absurdo2})

\item (c) $n_{1}\geq n_{0}$, where $n_{0}$ is defined in (\ref%
{propiedadOmega}).
\end{description}

By computing (a) $+(-1)$ {(b)}, we obtain 
\begin{equation*}
\tau (\mathcal{A}_{n_{1}}^{R_{0}},F_{n_{1}})-\tau (\mathcal{A}%
_{n_{1}},F_{n_{1}})\geq \nu _{0}^{(k-1)}-\nu _{0}-\frac{\varepsilon }{2}.
\end{equation*}%
Due to (c) $\varepsilon \geq \tau (\mathcal{A}_{n_{1}}^{R_{0}},F_{n_{1}})-%
\tau (\mathcal{A}_{n_{1}},F_{n_{1}})$, then 
\begin{equation*}
\varepsilon \geq \nu _{0}^{(k-1)}-\nu _{0}-\varepsilon /2\Rightarrow \frac{3%
}{2}\varepsilon \geq \nu _{0}^{(k-1)}-\nu _{0}.
\end{equation*}%
Noticing that $\varepsilon $ was chosen as $\varepsilon =(\nu
_{0}^{(k-1)}-\nu _{0})/2$ we get 
\begin{equation*}
\frac{3}{4}(\nu _{0}^{(k-1)}-\nu _{0})\geq (\nu _{0}^{(k-1)}-\nu
_{0})\Rightarrow \frac{3}{4}\geq 1,
\end{equation*}
which is an absurdity
\end{proof}

\subsection*{Proof of the main result (Strong Consistency)}

Now we can give a proof of Theorem 1.

\begin{proof}
\begin{enumerate}
\item \textbf{Construction of the Compact-set:} We take the $C$ constant and
the subset of probability $1$ from Lemma \ref{lemaCompacto1}, it is easy to
see that centers of $\mathcal{A}_{n}$ are inside a closed ball. Indeed, let $%
\boldsymbol{\mu}_n \in \mathcal{A}_{n}$ be an arbitrary center, by
definition of Hausdorff distance, exists $\boldsymbol{\mu} \in \mathcal{A}_0$%
, such that $\|\boldsymbol{\mu}_n-\boldsymbol{\mu}\| \leq d_H( \mathcal{A}%
_{n},\mathcal{A}_{0}) \leq C $, if $n>n_0$, then $\|\boldsymbol{\mu}_n\|
\leq \|\boldsymbol{\mu}\| + C$, taking maximum over $\|\boldsymbol{\mu}\|
\in \mathcal{A}_0$, we can find a properly upper bound:

\begin{equation*}
\|\boldsymbol{\mu}_n\| \leq C + \max_{\boldsymbol{\mu} \in \mathcal{A}_0} \|%
\boldsymbol{\mu} \| \quad \quad a.s
\end{equation*}

so that, we can assert, $\mathcal{A}_n \in \boldsymbol{\mathcal{K}}_0$,
being $\boldsymbol{\mathcal{K}}_0 = \{ \mathcal{A}\subseteq \overline{B}, \#%
\mathcal{A} \leq k\}$, where $\overline{B}$ is the closed ball the closed
ball centered on zero and radio equal to $C + \max_{\boldsymbol{\mu} \in 
\mathcal{A}_0} \|\boldsymbol{\mu} \|.$

\item \textbf{Convergence within the compact-set:} In order to prove
convergence, we apply Lemma \ref{funcCompactas} to $\varphi(\mathcal{A})=
\tau(\mathcal{A},F)$, indeed, we only must see that $\tau(\mathcal{A}_n,F)$
is close to $\tau(\mathcal{A}_0,F)$ in the sense given by Lemma \ref%
{funcCompactas}. Hence, we take $\tau(\mathcal{A}_n,F)$, then by adding and
subtracting $\tau(\mathcal{A}_n,F_n)$ we get, 
\begin{equation*}
\tau(\mathcal{A}_n,F) = \left[ \tau(\mathcal{A}_n,F) - \tau(\mathcal{A}%
_n,F_n) \right] + \tau(\mathcal{A}_n,F_n) \leq
\end{equation*}
\begin{equation*}
\leq \left[ \sup_{A\in \boldsymbol{\mathcal{K}_0}} \left|\tau(\mathcal{A},F)
- \tau(\mathcal{A},F_n)\right| \right] + \tau(\mathcal{A}_n,F_n)
\end{equation*}
Term between brackets is $o(1)$ 
because of the Uniform Strong Law of Larger numbers over the compact set $%
\boldsymbol{\mathcal{K}}_0$ from Lemma \ref{ConvUnifacop}. Adding and
substracting $\tau(\mathcal{A}_0,F_n)$ results 
\begin{equation*}
\tau(\mathcal{A}_n,F) \leq o(1) + \left[ \tau(\mathcal{A}_n,F_n) - \tau(%
\mathcal{A}_0,F_n)\right]+ \tau(\mathcal{A}_0,F_n),
\end{equation*}
the term between brackets is not positive, by definition of $\mathcal{A}_n$,
thus, 
\begin{equation*}
\tau(\mathcal{A}_n,F) \leq o(1) + \tau(\mathcal{A}_0,F_n)
\end{equation*}
The strong law of the large numbers implies that term $\tau(\mathcal{A}%
_0,F_n) = o(1) + \tau(\mathcal{A}_0,F)$. Therefore, we obtain 
\begin{equation*}
\tau(\mathcal{A}_n,F) \leq o(1) + \tau(\mathcal{A}_0,F)= o(1) + \nu_0,
\end{equation*}
hence, it is verified the hypothesis of Lemma \ref{funcCompactas}, which
leads to 
\begin{equation*}
d_H(\mathcal{A}_n, \mathcal{A}_0) \to 0
\end{equation*}
\end{enumerate}
\end{proof}

\end{document}